\documentclass[aps,prx,twocolumn,superscriptaddress,longbibliography,10pt]{revtex4-2}

\usepackage[colorlinks=true, allcolors=blue]{hyperref}
\usepackage[page,header]{appendix}
\usepackage{mathtools}

\usepackage{amsmath,amsfonts,amssymb,amstext, amsthm}
\usepackage[sort&compress]{natbib}
\usepackage{graphicx}
\usepackage{color}
\usepackage{comment}
\usepackage[normalem]{ulem}

\newtheorem{theorem}{Theorem}
\newtheorem{proposition}{Proposition}
\newtheorem{lemma}[theorem]{Lemma}
\newtheorem{corollary}{Corollary}

\theoremstyle{definition}
\newtheorem{definition}{Definition}

\theoremstyle{definition}
\newtheorem*{Proof}{Proof}

\newcommand{\ket}[1]{|#1\rangle} 
\newcommand{\bra}[1]{\langle#1|} 
\newcommand{\ketbra}[2]{|#1\rangle\langle#2|} 
\newcommand{\braket}[2]{\langle #1 \vert #2 \rangle}

\DeclareMathOperator{\tr}{Tr}

\newcommand{\id}{{\rm{id}}}

\newcommand{\cC}{{\mathcal{C}}}
\newcommand{\cD}{{\mathcal{D}}}
\newcommand{\cE}{{\mathcal{E}}}
\newcommand{\cF}{{\mathcal{F}}}
\newcommand{\cG}{{\mathcal{G}}}
\newcommand{\cH}{{\mathcal{H}}}

\newcommand{\cM}{{\mathcal{M}}}
\newcommand{\cN}{{\mathcal{N}}}

\newcommand{\cQ}{{\mathcal{Q}}}
\newcommand{\cR}{{\mathcal{R}}}

\newcommand{\cT}{{\mathcal{T}}}
\newcommand{\cU}{{\mathcal{U}}}
\newcommand{\cV}{{\mathcal{V}}}
\newcommand{\cW}{{\mathcal{W}}}

\newcommand{\mbb}[1]{\mathbb{#1}}
\newcommand{\EH}{\mbb{E}_{U \sim {\sf H}}}

\def\lsim{\mathrel{\rlap{\lower4pt\hbox{\hskip1pt$\sim$}}
		\raise1pt\hbox{$<$}}}                
\def\gsim{\mathrel{\rlap{\lower4pt\hbox{\hskip1pt$\sim$}}
		\raise1pt\hbox{$>$}}}                

\usepackage{thmtools,thm-restate}
\usepackage{regexpatch}
\makeatletter
\xpatchcmd\thmt@restatable{%
\csname #2\@xa\endcsname\ifx\@nx#1\@nx\else[{#1}]\fi
}{%
\ifthmt@thisistheone
\csname #2\@xa\endcsname\ifx\@nx#1\@nx\else[{#1}]\fi
\else
\csname #2\@xa\endcsname[{restated}]
\fi}{}{}
\makeatother

\usepackage[capitalize]{cleveref}

\begin{document}

        \begin{titlepage}
\thispagestyle{empty}

\begin{flushright}
YITP-22-120
\end{flushright}

\title{Decoding general error correcting codes and the role of complementarity}

\author{Yoshifumi Nakata}
\affiliation{Yukawa Institute for Theoretical Physics, Kyoto University, Oiwake-cho, Kitashirakawa, Sakyo-ku, Kyoto, 606-8502, Japan.}
\affiliation{Photon Science Center, Graduate School of Engineering, The University of Tokyo, Bunkyo-ku, Tokyo 113-8656, Japan}
\author{Takaya Matsuura}
\affiliation{RIKEN Center for Quantum Computing (RQC), Hirosawa  2-1, Wako, Saitama, 351-0198, Japan}
\affiliation{Centre for Quantum Computation \& Communication Technology,
School of Science, RMIT University, Melbourne, VIC 3000, Australia}
\affiliation{Department of Applied Physics, Graduate School of Engineering, The University of Tokyo, 7-3-1 Hongo, Bunkyo-ku, Tokyo 113-8656, Japan}
\author{Masato Koashi}
\affiliation{Photon Science Center, Graduate School of Engineering, The University of Tokyo, Bunkyo-ku, Tokyo 113-8656, Japan}
\affiliation{Department of Applied Physics, Graduate School of Engineering, The University of Tokyo, 7-3-1 Hongo, Bunkyo-ku, Tokyo 113-8656, Japan}
	
\begin{abstract}
Among various classes of quantum error correcting codes (QECCs), non-stabilizer codes have rich properties and are of theoretical and practical interest. Decoding non-stabilizer codes is, however, a highly non-trivial task. In this paper, we show that a decoding circuit for Calderbank-Shor-Steane (CSS) codes can be straightforwardly extended to handle general QECCs. 
The key to the extension lies in the use of a pair of classical-quantum (CQ) codes associated with the QECC to be decoded.
The decoding error of the proposed decoding circuit depends on the classical decoding errors of the CQ codes and their degree of complementarity. We demonstrate the power of the decoding circuit in a toy model of the black hole information paradox, improving decoding errors compared to previous results. In addition, we reveal that black hole dynamics may optimally encode quantum information but poorly encode classical information.
\end{abstract}

\maketitle
\end{titlepage}

\section{Introduction}
Reversing the effects of noise and recovering quantum information from a noisy quantum system are central to large-scale quantum information processing and also offer significant insight into fundamental physics~\cite{PSW2006,dRHRW2016,HP2007,SS2008,LSHOH2013,AH2015, PYHP2015,DHW2016,D2017,HQRY2016,RY2017,NWK2020}.
A common method to this end is quantum error correction (QEC), where quantum information is encoded into a system in such a way that it is decodable even after the system experiences noise.
The standard class of quantum error correcting codes (QECCs) is stabilizer codes, but there has been growing interest in QECCs beyond stabilizer codes due to their higher encoding rates, which can be capacity-achieving~\cite{L1997, S2002, D2005, BSST1999, BSST2002}. 
Moreover, non-stabilizer codes are significant in theoretical physics, as they play an important role in the recent studies of quantum many-body systems from the perspective of quantum information theory~\cite{HP2007,SS2008,LSHOH2013}.

An important class of non-stabilizer codes is random codes. They are originally introduced as an analytical tool for investigating the achievability of the quantum capacity. However, recent advances in the unitary design theory~\cite{DBWR2010,SDTR2013,BF2013,NHMW2017,WN2021,WN2020,NWY2021,GKHJF2021,DNTY2024} and quantum technologies have opened up the possibility of the practical use of random codes.
They also have numerous applications in theoretical physics, including the foundations of statistical mechanics~\cite{PSW2006,dRHRW2016}, scrambling dynamics in complex many-body systems~\cite{HP2007,SS2008,LSHOH2013}, and the study of quantum black holes~\cite{HQRY2016,RY2017,NWK2020}.

One of the obstacles in exploring non-stabilizer codes is the lack of explicit and efficient decoders.
For stabilizer codes, a decoder can be constructed by correcting errors based on syndrome measurements, followed by the application of the inverse of the encoding unitary.
This approach cannot be naively extended to non-stabilizer codes since measurements that preserve logical information are highly non-trivial, making it unclear how to perform syndrome measurements for non-stabilizer codes. Consequently, non-stabilizer codes are typically studied by the so-called \emph{decoupling} approach~\cite{HP2007, D2005,DW2004,DBWR2010}, which allows one to analyze code performance without explicitly constructing a decoder.

A few explicit decoders that are applicable to non-stabilizer codes have been proposed, though none are efficient. One such decoder is based on the Petz map~\cite{P1986}, which, despite its high computational cost~\cite{GLMQW2022}, is known to be useful for decoding general QECCs~\cite{BK2002,BDL2016}.
Another approach leverages an intrinsic relation between classical and quantum information~\cite{K2007, BR2009, T2010, Renes2016}, in which classical-quantum (CQ) codes associated to the QECC to be decoded play a central role. 
A CQ code, designed to protect only classical information from quantum noise, is obtained from any QECC by restricting its input to basis states, making the code basis-dependent, and by decoding it through quantum measurement.
In this approach, decoding measurements of two CQ codes associated to a given QECC, one defined by an arbitrary basis and the other by a complementary basis, are combined to construct a decoder for the original QECC.

In this paper, we explore the latter approach and show that, despite the lack of a priori stabilizer structures, a decoding quantum circuit for general QECCs can be constructed in a manner similar to stabilizer codes. 
To do so, we first provide a decoding circuit for a Calderbank-Shor-Steane (CSS) code~\cite{CS1996,S1996} by combining one-bit teleportation~\cite{ZLC2000} with the classical decoders that the CSS code inherently possesses. 
We then extend this circuit to general QECCs by replacing the classical decoders with two decoding measurements of the CQ codes associated with the QECC to be decoded, thereby making the circuit applicable to general QECCs.
We refer to the extended decoding circuit as a \emph{Classical-to-Quantum (C-to-Q) decoder} and show that the decoding error is determined by the decoding errors of the CQ codes and the degree of complementarity of the two bases of the CQ codes.
We also show that, with a suitable choice of decoding measurements of the CQ codes, the C-to-Q decoder is nearly optimal and can be used for decoding capacity-achieving QECCs.

We further demonstrate the power of the C-to-Q decoder by applying it to the Hayden-Preskill protocol~\cite{HP2007}, a toy model for the black hole information paradox that has been experimentally realized~\cite{LFSLYYM2019}.
By combining the C-to-Q decoder with a slightly modified pretty-good-measurement (PGM)~\cite{HW1994}, we compute the decoding errors for both classical and quantum information. As a result, we find that the C-to-Q decoder achieves the quantum capacity of the protocol and improves upon previous analyses based on the decoupling approach. 
Furthermore, we show that the black hole dynamics optimally encode quantum information but poorly encode classical information.

The C-to-Q decoder offers a wide range of advantages, both theoretical and practical. Theoretically, it not only contributes to the development of explicit decoders for general QECCs but also quantitatively reveals the fundamental role of complementarity in decoding QECCs. On the practical side, the C-to-Q decoder has multiple applications.
Even for stabilizer codes, the C-to-Q decoder can handle less restrictive scenarios, as it does not rely on syndrome measurements that must be fully commutative. Additionally, the C-to-Q decoder can be used for experimentally benchmarking the performance of a QECC by evaluating the corresponding CQ codes, simplifying the process of performance evaluation.
Moreover, the C-to-Q decoder is beneficial for simultaneously achieving error correction and code switching, which has promising applications in quantum communication, as will be discussed later.

This paper is organized as follows. We begin with preliminaries in~\cref{S:SettingResults}. All results are summarized in~\cref{S:MR}. Proofs of the main statements about the C-to-Q decoder are provided in~\cref{S:CtoQProof}. An in-depth analysis of the Hayden-Preskill protocol with the C-to-Q decoder is presented in~\cref{S:ErrorHP}.
Proofs of technical statements are given in Appendices after the summary and discussion in~\cref{S:COs}.

\section{Preliminaries} \label{S:SettingResults}

\subsection{Notation}

Throughout the paper, we use superscripts to refer to the relevant systems, such as a Hilbert space $\cH^A$ of a system $A$, an operator $X^{AB}$ acting on $AB$, and a superoperator $\cE^{A \rightarrow B}$ from $A$ to $B$. 
A superoperator from $A$ to itself is denoted by ${\cal E}^A$. The superscript is sometimes omitted when it is clear from the context.

When we say a basis of a Hilbert space, we always mean orthonormal ones.
Let $\{ \ket{e_j} \}_j$ be a basis in $A$, and $A'$ be the system isomorphic to $A$. The state 
$\ket{\Phi}^{AA'} : = (\dim \cH^A)^{-1/2} \sum_j \ket{e_j}^A \otimes \ket{e_j}^{A'}$ is called a maximally entangled state. In terms of another basis $\{ \ket{f_j} \}_j$, the state $\ket{\Phi}^{AA'}$ is given by
\begin{equation}
\ket{\Phi}^{AA'} = (\dim \cH^A)^{-1/2} \sum_j \ket{f_j}^A \otimes \ket{f_j^*}^{A'},
\end{equation}
where $\ket{f_j^*}^{A'} = \sqrt{\dim \cH^A}\bra{f_j}^A \ket{\Phi}^{AA'}$. The corresponding density matrix $\ketbra{\Phi}{\Phi}^{AA'}$ is denoted by $\Phi^{AA'}$.
We denote the completely mixed state on $A$ by $\pi^A := I^A/\dim \cH^A$, where $I^A$ is the identity operator on $A$.

The Schatten $p$-norm for a linear operator $X$ is defined by $\| X \|_p := ( \tr [ (X^{\dagger}X)^{p/2}] )^{1/p}$ ($p \in [1, \infty]$). We particularly use the trace ($p=1$), and operator ($p = \infty$) norms. The fidelity between quantum states $\rho$ and $\sigma$ is defined by $F(\rho, \sigma) := \| \sqrt{\rho} \sqrt{\sigma} \|_1^2$. The fidelity and the trace norm are related to each other by the Fuchs–van de Graaf inequalities as
\begin{equation}
    1 - \sqrt{F(\rho, \sigma)} \leq \frac{1}{2} \bigl\| \rho - \sigma \bigr\|_1 \leq \sqrt{ 1 - F(\rho, \sigma)}. \label{FvG}
\end{equation}
For any state $\rho$, which may be defined on a composite system including $A$, the collision entropy $H_2(A)_{\rho}$ is defined by
\begin{equation}
    H_2(A)_{\rho} = -\log \bigl[ \tr (\rho^A)^2 \bigr].
\end{equation}

Important classes of dynamics in quantum systems are \emph{unitary}, \emph{isometry}, and \emph{partial isometry}. An isometry $V^{A \rightarrow A'}$ from a system $A$ to $A'$ ($\dim \cH^A \leq \dim \cH^{A'}$) is a linear operator such that $V^{A \rightarrow A'\dagger} V^{A \rightarrow A'} = I^{A}$. When $\dim \cH^A = \dim \cH^{A'}$, it holds that $V^{A \rightarrow A'} V^{A \rightarrow A' \dagger} = I^{A}$ as well, so the isometry is unitary. 
A linear operator $W^{A \rightarrow A'}$ is a partial isometry if both $W^{A \rightarrow A' \dagger} W^{A \rightarrow A'}$ and $W^{A \rightarrow A'} W^{A \rightarrow A
' \dagger}$ are projections, the former on $A$ and the latter on $A'$.
Projections, isometries, and unitaries are special classes of partial isometries.

A \emph{quantum channel} $\cT^{A \rightarrow B}$ is a completely-positive (CP) and trace-preserving (TP) map. A map is called CP if $({\rm id}^{A'} \otimes \cT^{A \rightarrow B}) (\rho^{A'A}) \geq 0$ for any $\rho^{A'A} \geq 0$, where ${\rm id}$ is the identity map, and is TP if $\tr[\cT^{A \rightarrow B} (\rho^{A})] = \tr[\rho^A]$. 
For simplicity, we denote $({\rm id}^{A'} \otimes \cT^{A \rightarrow B} ) (\rho^{A'A})$ by $\cT^{A \rightarrow B}(\rho^{A'A})$.
For a given partial isometry, we sometimes denote the conjugating map by its calligraphic style, such as $\cV(\rho) := V \rho V^{\dagger}$ for a partial isometry $V$.
A fundamental quantum channel is the partial trace over a subsystem that is indicated by a subscript, such as $\tr_C$. In this paper, a partial trace is often implicit in the sense that $\rho^A$ denotes the reduced operator on $A$ of $\rho^{AB}$, i.e., $\rho^A = \tr_B [\rho^{AB}]$.

For a superoperator $\cM^{A \rightarrow B}$, the diamond norm is defined as 
\begin{equation}
    \| \cM^{A \rightarrow B} \|_{\diamond} := \max_{\cH^R}\max_{O^{AR} (\neq 0)} \frac{\| \cM^{A \rightarrow B} \otimes \id^R(O^{AR}) \|_1}{\| O^{AR} \|_1},
\end{equation}
where $\max_{\cH^R}$ is taken over a Hilbert space $\cH^R$ with arbitrary finite dimension, $\id^R$ is the identity map on $R$, and $O^{AR}$ is any operator on $\cH^A \otimes \cH^R$. It suffices to take the dimension of $R$ at most the dimension of $A$~\cite{Watrous2018}.
Furthermore, when the map is Hermitian-preserving, the diamond norm is given by
\begin{equation}
    \| \cM^{A \rightarrow B} \|_{\diamond} = \max_{\cH^R}\max_{\ket{\psi}^{AR}} \| \cM^{A \rightarrow B} \otimes \id^R(\ketbra{\psi}{\psi}^{AR}) \|_1,\label{Def:DiamondNorm}
\end{equation}
where $\max_{\ket{\psi}^{AR}}$ is taken over all pure states in $\cH^{AR}$~\cite{Watrous2018}.

The Haar measure ${\sf H}$ on a unitary group $\mathfrak{U}(d)$ of finite degree $d$ is the unique left- and right- unitarily invariant probability measure. That is, for any measurable set $\cW \subset \mathfrak{U}(d)$ and $V \in \mathfrak{U}(d)$,
\begin{equation}
{\sf H}(V \cW) = {\sf H} (\cW V) = {\sf H}(\cW),
\end{equation}
and ${\sf H}(\mathfrak{U}(d)) = 1$.
When a unitary $U$ is chosen uniformly at random with respect to the Haar measure ${\sf H}$, it is called a Haar random unitary and is denoted by $U \sim {\sf H}$. 
The average of a function $f(U)$ over a Haar random unitary is denoted by $\mathbb{E}_{U \sim {\sf H}}[f(U)]$.

\subsection{Quantum error correcting codes (QECCs) and classical-quantum (CQ) codes}

A QECC for a noisy quantum channel $\cN^{B \rightarrow C}$ is defined by a pair of encoding and decoding quantum channels, $(\cE^{A \rightarrow B}, \cD^{C \rightarrow A})$, that satisfies
\begin{equation}
    \cD^{C \rightarrow A} \circ \cN^{B \rightarrow C} \circ \cE^{A \rightarrow B} \approx {\rm id}^A, \label{Eq:DiamondPrimitive}
\end{equation}
where ${\rm id}^A$ is the identity map on $A$. 
As a QECC is designed to achieve the identity map, it can preserve any quantum state on $A$, which may be correlated with other systems, even if the system $A$ experiences noise represented by $\cN$.  In this sense, a QECC is for protecting quantum information.

In contrast, a CQ code is designed to protect only classical information. A standard formulation of a CQ code for $(\log d)$-bit information is given by a set of $d$ states, $\{\rho_j\}_{j =0}^{d-1}$, and quantum measurement, i.e., positive-operator-valued measure (POVM), with $d$ outcomes $M=\{ M_j \}_{j=0}^{d-1}$. The CQ code works as follows: when the classical information to be protected is $j$, the corresponding quantum state $\rho_j$ is chosen. After the state experiences noise $\cN$, it is decoded by the measurement, which is successful if the measurement outcome is $j$, i.e., if $\tr[M_j \cN(\rho_j)] \approx 1$.

In this paper, as we focus on the connection between QECCs and CQ codes, we formulate the CQ code by explicitly describing the preparation process of $\rho_j$. More concretely, we consider that the state $\rho_j$ is prepared by applying an encoding channel $\cE$ to $\ket{j_W}$, one of the fixed basis states $W=\{ \ket{j_W} \}_{j =0}^{d-1}$, which encodes the $(\log d)$-bit classical information $j$.
The noisy state is, hence, given by
\begin{equation}
    \cN^{B \rightarrow C} \circ \cE^{A \rightarrow B}(\ketbra{j_W}{j_W}^A),
\end{equation}
which is decoded by a decoding measurement $M=\{ M_j^C \}_{j=0}^{d-1}$. The condition on succeeding in decoding the CQ code is as follows:
\begin{equation}
    \tr\bigl[ M_j^C \cN^{B \rightarrow C} \circ \cE^{A \rightarrow B}(\ketbra{j_W}{j_W}^A) \bigr] \approx 1, \label{Eq:CQerrorPrimitive}
\end{equation}
for all $j =0, \dots, d-1$. 

In this way, we formulate a CQ code by a tuple $(W, \cE, M)$, where $W$ is a fixed basis identifying the input state, $\cE$ is an encoding channel, and $M$ is the decoding measurement.
It is important to note that this formulation of the CQ codes depends on the choice of the basis $W$. Various CQ codes can be defined for the same encoding channel $\cE$.

\subsection{Decoding errors}

We here explain how the decoding errors of the QECCs and CQ codes are quantified.

\subsubsection{Decoding error of quantum information}

The decoding error of a QECC $(\cE, \cD)$ for a noisy quantum channel $\cN$ is defined based on~\cref{Eq:DiamondPrimitive}.
A natural way to quantify the approximation is to use the diamond norm:\begin{equation}
    \Delta_{\diamond}(\cD|\cN \circ \cE) := \frac{1}{2} \bigl\| \id^{A} - \cD^{C \rightarrow A} \circ \cN^{B \rightarrow C} \circ \cE^{A \rightarrow B} \bigr\|_{\diamond}.
\end{equation}
Another standard way is to introduce a reference system $R$, where $\dim \cH^A = \dim \cH^R = d$, and to use the maximally entangled state $\Phi^{AR}$:
\begin{equation}
    \Delta_{\rm{q}}(\cD|\cN \circ \cE) := \frac{1}{2} \bigl\| \Phi^{AR} - \cD^{C \rightarrow A} \circ \cN^{B \rightarrow C} \circ \cE^{A \rightarrow B}(\Phi^{AR}) \bigr\|_1. \label{Eq:definitionDecodingErrorQ}
\end{equation}

These two definitions of errors are closely related since they satisfy (see, e.g.,~\cite{KR2021})
\begin{equation}
    \Delta_{\rm{q}}(\cD|\cN \circ \cE) \leq \Delta_{\diamond}(\cD|\cN \circ \cE) \leq d \Delta_{\rm{q}}(\cD|\cN \circ \cE). \label{Eq:TwoErrors}
\end{equation}
Hence, when $\Delta_{\rm{q}}(\cD|\cN \circ \cE)$ is sufficiently small, so is $\Delta_{\diamond}(\cD|\cN \circ \cE)$.
From~\cref{Eq:TwoErrors}, we observe that $\Delta_{\rm{q}}(\cD|\cN \circ \cE)$ and $\Delta_{\diamond}(\cD|\cN \circ \cE)$ can differ by factor $d$, which may be problematic for large $d$. This issue can be circumvented in the context of QEC. In fact, we can improve~\cref{Eq:TwoErrors} by slightly modifying the encoding and decoding operations from the original QECC $(\cE, \cD)$. Here, we provide two methods of doing so.

The first method is to use a random unitary in the encoding and decoding operations. We especially use a unitary $1$-design, which is a random unitary that has the same first order moment as that of a Haar random unitary. A canonical instance of a unitary $1$-design is the multi-qubit Pauli group. Using a unitary $1$-design, the following statement holds. See~\cref{App:MESdiamondUnitary} for the derivation. 

\begin{restatable}{proposition}{MESdiamondUnitary}\label{Prop:MESdiamondUnitary}
    Let $(\cE^{A \rightarrow B}, \cD^{C \rightarrow A})$ be a QECC for a noisy quantum channel $\cN^{B\rightarrow C}$. Let $U^A$ be a unitary $1$-design, and define a new QECC $(\cE_U^{A \rightarrow B}, \cD^{C\rightarrow A}_{U^{\dagger}})$ by
    \begin{align}
        &\cE_U^{A \rightarrow B} (\rho^{A}) := \cE^{A \rightarrow B}(U^A \rho^A U^{A \dagger}),\\
        &\cD_{U^{\dagger}}^{C \rightarrow A} (\sigma^{C}) := U^{A \dagger}\cD^{C \rightarrow A}(\sigma^C) U^A.
    \end{align}
    Then, it holds that
    \begin{multline}
        \frac{1}{2} \bigl\| \id^{A} - \mbb{E}_U\bigl[ \cD_{U^{\dagger}}^{C \rightarrow A} \circ \cN^{B \rightarrow C} \circ \cE_{U}^{A \rightarrow B}\bigr] \bigr\|_{\diamond} \\
        \leq \Delta_{\rm{q}}(\cD|\cN \circ \cE),
    \end{multline}
    where $\mbb{E}_U$ is the average over the unitary $1$-design $U^A$.
\end{restatable}

\cref{Prop:MESdiamondUnitary} implies that if there exists a QECC $(\cE^{A \rightarrow B}, \cD^{C \rightarrow A})$ that achieves a small decoding error for a noisy quantum channel $\cN^{B \rightarrow C}$ in terms of $\Delta_{\rm{q}}(\cD|\cN \circ \cE)$, then one can achieve the same error in terms of the diamond norm by applying a random unitary $U^A$ and $U^{A \dagger}$ before and after the encoding and decoding, respectively. Hence, if $\Delta_{\rm{q}}(\cD|\cN \circ \cE)$ is small, it is possible to correct the noise $\cN^{B\rightarrow C}$ with the same decoding error in the diamond norm.

In the above scheme, the encoder and the decoder have to share common randomness to sample the same instance of the unitary 1-design. If one would like to avoid using common randomness, the following second method can be used. See, e.g.,~\cite{Watrous2018} for the proof.

\begin{proposition}\label{Prop:MESdiamond}
    Let $(\cE^{A \rightarrow B}, \cD^{C \rightarrow A})$ be a QECC for a noisy quantum channel $\cN^{B\rightarrow C}$. For any subspace $\cH^{A_0} \subseteq \cH^A$ with $\dim \cH^{A_0} \leq d/2$, there exists a pair of quantum channels $\cF^{A_0 \rightarrow A}$ and $\cG^{A \rightarrow A_0}$ such that a new QECC $(\cE \circ \cF, \cG \circ \cD)$ satisfies
    \begin{equation}
        \Delta_{\diamond}(\cG \circ \cD|\cN \circ \cE \circ \cF) \leq 2\sqrt{2\Delta_{\rm q}(\cD|\cN\circ \cE)}.
    \end{equation}
\end{proposition}

We observe from~\cref{Prop:MESdiamond} that, if $\Delta_{\rm{q}}(\cD|\cN\circ \cE) \ll 1$, a small error can be achieved in terms of the diamond norm by restricting the $d$-dimensional Hilbert space $\cH^A$ to the one with a dimension at most $d/2$. This corresponds to reducing one qubit of logical quantum information. Hence, this relation shows that one can relate the decoding error $\Delta_{\rm q}$ with $\Delta_{\diamond}$ at the expense of one logical qubit.

For these reasons, the decoding error for quantum information can be characterized well by $\Delta_{\rm q}(\cD|\cN\circ \cE)$ in QEC.
In addition, as we will show below, the decoding error $\Delta_{\rm q}(\cD|\cN \circ \cE)$ is directly related to the decoding errors of the associated CQ codes. Thus, we use the decoding error $\Delta_{\rm{q}}(\cD|\cN \circ \cE)$ in this paper.

\subsubsection{Decoding error of classical information}

For a CQ code $(W,\cE, M)$, the decoding error is defined by the failure probability of the decoding measurement.
Taking the uniform average over all possible classical inputs $j = 0,\dots, d-1$, we define the average decoding error of classical information as
\begin{multline}
    \Delta_{{\rm cl}, W}(M|\cN \circ \cE) \\
    := 
    \frac{1}{d} \sum_{i \neq j}
    \tr\bigl[M_i^C \cN^{B \rightarrow C} \circ \cE^{A \rightarrow B} (\ketbra{j_W}{j_W}^A)   \bigr].\label{Eq:REclassical}
\end{multline}
In the following, we refer to this decoding error as the classical decoding error, especially when we need to distinguish it clearly from the decoding error for quantum information.

The classical decoding error $\Delta_{{\rm cl}, W}(M|\cN \circ \cE)$ can be rephrased in terms of the maximally-correlated $W$-classical state $\Omega^{AR}_W$ defined by
\begin{equation}
    \Omega_W^{AR} := \frac{1}{d} \sum_{j} \ketbra{j_W}{j_W}^A \otimes \ketbra{j_W}{j_W}^{R}.
\end{equation}
To see this, we introduce a quantum channel $\cD_{M}^{C \rightarrow A}(\rho^C) = \sum_{j =0}^{d-1} \tr\bigl[  M_j^C \rho^C \bigr] \ketbra{j_W}{j_W}^A$. Using this, we have
\begin{multline}
    \Delta_{{\rm cl}, W}(M|\cN \circ \cE) \\
    = \frac{1}{2} \bigl\| \Omega_W^{AR} - \cD_M^{C \rightarrow A} \circ \cN^{B \rightarrow C} \circ \cE^{A \rightarrow B}(\Omega_W^{AR}) \bigr\|_1.
\end{multline}
    
Note that optimal decoding error for classical information never exceeds that for quantum information. This is due to the facts that $\Omega_W^{AR} = \cC^R_W(\Phi^{AR})$, where $\cC_W$ is the completely dephasing channel in the $W$-basis, and that $\cC_W\circ \cD$ is equivalent to a map $\cD_M$ with suitably chosen POVM $M$. Using the monotonicity of the trace norm, we obtain that, for any $\cN \circ \cE$ and $\cD$, there exists a POVM $M$ such that $\Delta_{{\rm cl}, W}(M|\cN \circ \cE)  \leq \Delta_{\rm{q}}(\cD|\cN \circ \cE)$.

\section{Main results} \label{S:MR}

We here summarize our results. 
We present a construction of the C-to-Q decoder and its performance analysis in~\cref{SS:MT} . For completeness, we discuss a possible choice of decoding measurements for the CQ codes used in the construction in~\cref{SS:pPGMs}. Finally, our results on the decoding problem in the Hayden-Preskill protocol are summarized in~\cref{SS:HPCtoQ}.

\subsection{The Classical-to-Quantum decoder}  \label{SS:MT}

We begin with a decoding quantum circuit for a CSS code in~\cref{SSS:DecodeCSS}, which is constructed by combining one-bit teleportation with the classical decoders for the CSS code.
In~\cref{SSS:a}, we extend the decoding circuit to general QECCs and provide a construction of the C-to-Q decoder. A few implications of the C-to-Q decoder are elaborated on in~\cref{SSS:b}.

\subsubsection{A decoding circuit for CSS codes} \label{SSS:DecodeCSS}
Let us consider an $[[N,k]]$-CSS code that encodes $k$ logical qubits into $N$ physical qubits.
The code is constructed from two classical linear codes $C_1$ and $C_2$ on $N$ bits that satisfy $C_2 \subset C_1$. The difference between the number of logical bits for $C_1$ and that for $C_2$ should be $k$. 
In the CSS code, the logical Pauli-$Z$ basis is defined by the classical code $C_1/C_2$, and the logical Pauli-$X$ basis is by $C_2^{\perp}/C_1^{\perp}$.
Let $f_1$ and $f_2$ be the classical decoders for $C_1$ and $C_2^{\perp}$, respectively. They are both boolean functions, taking $N$ bits as an input and outputting $k$ bits. In the CSS code, bit- and phase-flip errors are independently corrected by using $f_1$ and $f_2$, respectively. Note that correcting bit- and phase-flip errors suffice for completing error correction. See~\cref{App:CSS} for the details.

To decode a CSS code after error correction, one can, for instance, simply apply the inverse of the encoding circuit. However, performing error correction and then decoding the code is not the unique way.
In fact, using the fact that the bit- and phase-flip errors are independently corrected in a CSS code, we can construct a decoding quantum circuit that simultaneously achieves error correction and decoding. The idea is based on one-bit teleportation protocol~\cite{ZLC2000}.

\begin{figure}[tb!]
\centering
\includegraphics[width=.3\textwidth,clip]{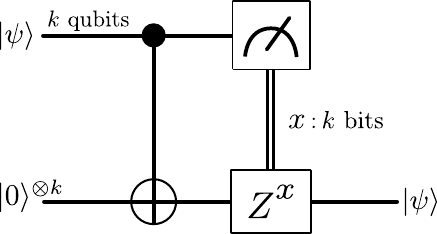}
\caption{A multi-qubit extension of the one-bit teleportation. A $k$-qubit quantum state can be transferred to an ancillary system by applying the controlled-NOT gates, measuring the $k$ qubits in the Pauli-$X$ basis, and applying the feed-forward unitary $Z^x$ to the ancillary system, which depends on the measurement outcome $x \in \{0,1\}^k$.
}
\label{Fig:OBT}
\end{figure}

\begin{figure}[tb!]
\centering
\includegraphics[width=.4\textwidth,clip]{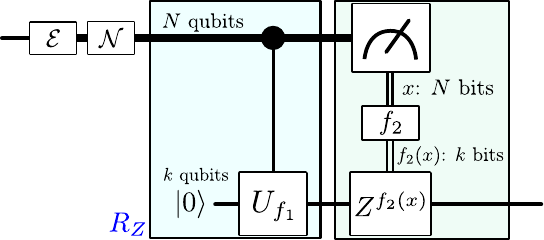}
\caption{A decoding quantum circuit for an $[[N,k]]$-CSS code, which transforms a noisy $N$-qubit state to a $k$-qubit state. The decoder is constructed from the two classical decoders $f_1$ and $f_2$ of the classical linear codes that define the CSS code. Note that the isometry $R_Z$, given by~\cref{eq:z_basis_isometry}, is depicted by attaching ancillary qubits and applying the multi-qubit controlled unitary.
}
\label{Fig:CSSCodes}
\end{figure}

One-bit teleportation is a protocol that transfers a state to an ancillary system in the noiseless situation, and is equivalent to the identity gate in the measurement-based quantum computation~\cite{RB2001,RBB2003}.
In~\cref{Fig:OBT}, a $k$-qubit state $\ket{\psi}$ is transferred to a $k$-qubit ancillary system by the following two steps. 
First, a tensor product of the controlled-NOT gates is applied. Second, the input of $k$ qubits are measured in the Pauli-$X$ basis, and a feed-forward unitary $Z^x = Z^{x_1} \otimes \dots \otimes Z^{x_k}$ is applied to the ancillary system, where $x = (x_1, \dots, x_k)$ is the measurement outcome.

To construct a decoder based on the one-bit teleportation, it is important to notice that the first and the second steps of the one-bit teleportation are based only on the $Z$- and $X$-basis properties of the input state, respectively. 
This feature is well-suited to a CSS code since, in a CSS code, the classical decoder $f_1$ ($f_2$), which is used for correcting the bit-flip (phase-flip) error, is based only on the $Z$-basis ($X$-basis) properties of the noisy state. 
In fact, as we demonstrate below, a decoding circuit for a CSS code is obtained simply by plugging the classical decoders $f_1$ and $f_2$ into the first and the second step of the one-bit teleportation, respectively. 

Similarly to the first step of one-bit teleportation, a controlled unitary is applied to the input of the decoder and the ancillary system, where the former controls the latter.
The controlled unitary here is such that, if the input system is in $\ket{z}$, a unitary $U_{f_1}$ satisfying $\ket{0}^{\otimes k} \mapsto \ket{f_1(z)}$ is applied to the ancillary system.
One may rephrase this operation, including the preparation of the ancillary system, by an isometry $R_Z$ given by
\begin{equation}
    R_Z = \sum_{z \in \{0, 1\}^N} \ketbra{z}{z} \otimes \ket{f_1(z)}. \label{eq:z_basis_isometry}
\end{equation}
By this operation, a \emph{decoded} bit string $f_1(z) \in \{0, 1\}^k$ in the Pauli-$Z$ basis is coherently copied to the $k$-qubit ancillary system. If the decoding error of the classical decoder $f_1$ is sufficiently small, the decoded bit string $f_1(z)$ does not contain any bit-flip error.

The second step is also similar to one-bit teleportation, but the feed-forward unitary is based on the outcome of the classical decoder $f_2$. Let $x \in \{0, 1\}^N$ be the measurement outcome of the input system, when it is measured in the Pauli-$X$ basis. The feed-forward unitary onto the ancillary system is then given by $Z^{f_2(x)}$, where $f_2(x)$ is the \emph{decoded} $k$-bit string. As the classical decoder $f_2$ in the CSS code is used for correcting the phase-flip error, if the decoding error of $f_2$ is sufficiently small, the $k$ bits for determining the feed-forward unitary are correctly chosen.

\begin{figure*}[tb!]
\centering
\includegraphics[width=.8\textwidth,clip]{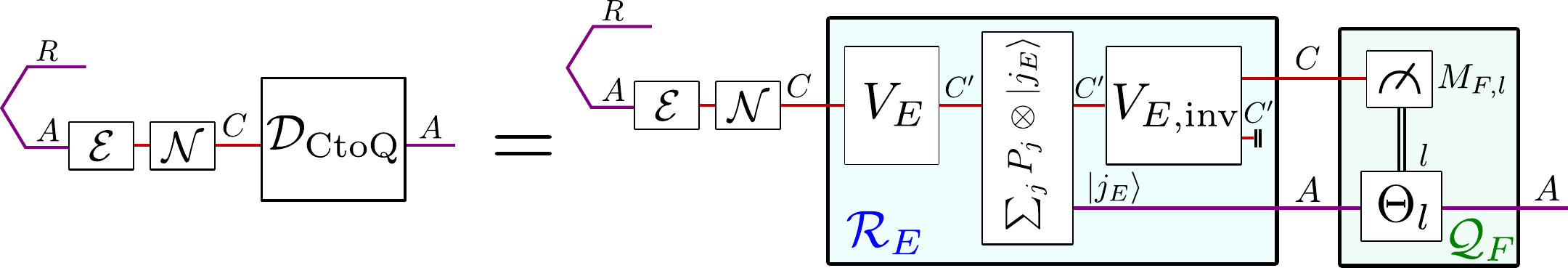}
\caption{Construction of the Classical-to-Quantum decoder from POVMs $M_E$ and $M_F$, which consists of two quantum channels $\cR_E$ and $\cQ_F$. The former $\cR_E$ is an isometry for coherently recording the measurement outcome by the POVM $M_E$ into an ancillary system $A$. The $(V_E, \{ P_j\})$ is a Naimark extension of $M_E$. The isometry $V_{E, \textrm{inv}}$ acts as if it is an inverse of the isometry $V_E$, which is for minimizing the back-action to $C$. 
If the decoding error $\Delta_{{\rm cl}, E}(M_E|\cN \circ \cE)$ is sufficiently small, the channel $\cR_E$ transforms the state on $RCA$ to a noisy GHZ state.
The latter quantum channel $\cQ_F$ plays the role of quantum eraser, which aims to delete all the $E$-classical information from $C$. The eraser succeeds if $(E, F)$ is close to MUB and $\Delta_{{\rm cl}, F}(M_F|\cN \circ \cE)$ is small, transforming the noisy GHZ state in $RCA$ to a maximally entangled state between $RA$.
}
\label{Fig:decoder}
\end{figure*}

In total, the decoding circuit for CSS codes consists of two steps: one is a coherent use of one classical decoder $f_1$, corresponding to the isometry $R_Z$ in~\cref{eq:z_basis_isometry}, and the other is the measurement and feed-forward based on the other classical decoder $f_2$.
The complete decoding circuit is provided in~\cref{Fig:CSSCodes}.
Through these two steps, error correction is completed while the quantum state is transferred to the ancillary system. In other words, error correction and decoding is simultaneously achieved. 

To understand how the circuit works, it is pedagogical to consider the noiseless case, which is nothing but one-bit teleportation from a logical system to an ancillary system.
Let $\ket{\bar{\psi}} = \sum_{j \in \{0,1\}^k} \alpha_j \ket{\bar{j}_Z}$ be an encoded logical state, where $\ket{\bar{j}_Z}$ is the logical Pauli-$Z$ basis in the $N$-qubit system.
With the absence of noise, the controlled unitary with the classical decoder $f_1$ simply copies $j_Z$ to the ancillary system, leading to the transformation such as
\begin{equation}
    \ket{\bar{\psi}}=\sum_{j \in \{0,1\}^k} \alpha_j \ket{\bar{j}_Z} \mapsto \sum_{j \in \{0,1\}^k} \alpha_j \ket{\bar{j}_Z} \otimes \ket{j_Z}.
\end{equation}
If the information of $j_z$ is erased from the input system, we obtain $\ket{\psi}$ in the ancillary system, completing the decoding task. This can be achieved by measuring the input system in the logical Pauli-$X$ basis and by applying the feed-forward unitary on the ancillary system. This is exactly what the second step with the classical decoder $f_2$ does. As a result, the state transformation
\begin{equation}
    \sum_{j \in \{0,1\}^k} \alpha_j \ket{\bar{j}_Z} \otimes \ket{j_Z}\mapsto 
    \sum_{j \in \{0,1\}^k} \alpha_j \ket{j_Z} = \ket{\psi},
\end{equation}
is realized, and a decoded state is obtained in the ancillary system.

The situation does not change much in the presence of noise as far as the noise is correctable by the CSS code. This is simply because, in each step, the classical decoders $f_1$ and $f_2$ correct the bit- and phase-flip errors, respectively.
In~\cref{App:CSS}, we provide an in-depth analysis on the decoding error $\Delta_{\rm q}$ for the quantum information by this decoding circuit and show that
\begin{equation}
    \Delta_{\rm q} \leq \sqrt{\Delta_{{\rm cl}, Z} + \Delta_{{\rm cl}, X}},
    \label{eq:error_for_css}
\end{equation}
where $\Delta_{{\rm cl}, Z}$ and $\Delta_{{\rm cl}, X}$ are the classical decoding errors of the classical decoders $f_1$ and $f_2$, respectively.

Note that $\Delta_{{\rm cl}, Z}$ and $\Delta_{{\rm cl},X}$ are equivalent to the classical decoding errors of CQ codes associated to the CSS code, defined in the $Z$- and $X$-bases, respectively.
More specifically, denoting the encoding map of a CSS code by $\cE$ and the noisy quantum channel by $\cN$, the classical decoding error $\Delta_{{\rm cl}, W}$ for $W=X, Z$ is given by
\begin{equation}
    \Delta_{{\rm cl}, W} = \Delta_{{\rm cl}, W}( M_W | \cN \circ \cE ),
\end{equation}
where $M_W$ is the measurement followed by the classical decoding process either $f_1$ or $f_2$.

\subsubsection{Construction of the C-to-Q decoder} \label{SSS:a}

The main idea of the C-to-Q decoder is to follow the above construction of the decoding circuit for CSS codes.
We, however, use decoding measurements of CQ codes instead of the classical decoders in each step.

Given an encoding channel $\cE$ of a QECC, let us consider two CQ codes $(E, \cE, M_E)$ and $(F, \cE, M_F)$. Here, $E := \{ \ket{j_E} \}_{j=0}^{d-1}$ and $F := \{ \ket{l_F} \}_{l=0}^{d-1}$ are bases, and $M_W:= \{ M^C_{W, j} \}_{j= 0}^{d-1}$ ($W =E,F$) are decoding measurements. The decoding errors for classical information are $\Delta_{{\rm cl}, W}(M_W|\cN\circ \cE)$.

A C-to-Q decoder $\cD^{C \rightarrow A}_{\rm CtoQ}$ based on the decoding measurements $M_E$ and $M_F$ of the two CQ codes is constructed as in~\cref{Fig:decoder}.
Similarly to the decoding circuit for CSS codes in~\cref{Fig:CSSCodes}, the C-to-Q decoder consists of two steps represented by quantum channels $\cR_E^{C\rightarrow CA}$ and $\cQ_F^{CA \rightarrow A}$.
The first quantum channel $\cR_E^{C\rightarrow CA}$ corresponds to the isometry for coherently recording the outcome of one decoding measurement $M_E$ into an ancillary system, and the second one $\cQ_F^{CA \rightarrow A}$ is based on the other decoding measurement $M_F$ and plays a role of the measurement and feed-forward.

The first quantum channel $\cR_E^{C\rightarrow CA}$ is composed of an isometry $R_E^{C \rightarrow C C' A}$ from $C$ to $CC'A$. To define the isometry $R_E^{C \rightarrow C C' A}$, let $(V_E^{C \rightarrow C'},\{P^{C'}_j\}_{j=0}^{d-1})$ be a Naimark extension of the POVM $M_E$, that is, a pair of an isometry $V_E^{C\rightarrow C'}$ and orthogonal projections $\{ P^{C'}_j \}_{j=0}^{d-1}$ such that
\begin{equation}
    M^C_{E,j}=(V_E^{C\rightarrow C'})^{\dagger} P^{C'}_j V_E^{C\rightarrow C'},
\end{equation}
where $\sum_{j=0}^{d-1} P_j^{C'} = I^{C'}$.
Then, 
\begin{equation}
R_E^{C \rightarrow C C' A}  := V_{E, \textrm{inv}}^{C'\rightarrow CC'}  \biggl(\sum_{j = 0}^{d-1} P_j^{C'} \otimes \ket{j_E}^{A} \biggr)V_E^{C\rightarrow C'},\label{Eq:cohmeas1}
\end{equation}
with 
\begin{multline}
    V_{E, \textrm{inv}}^{C'\rightarrow CC'} :=
    V_E^{C\rightarrow C'\dagger} \otimes  \ket{e_0}^{C'}  \\
    + \ket{e'_0}^{C} \otimes (I^{C'} - V_E^{C\rightarrow C'}V_E^{C\rightarrow C' \dagger}),\label{Eq:unitarize}
\end{multline}
where $\ket{e_0}^{C'}$ is a unit vector in the range of $V_E^{C\rightarrow C'}$, and $\ket{e'_0}^{C}$ is an arbitrary unit vector in $C$. 
As $V_E^{C\rightarrow C'}V_E^{C\rightarrow C' \dagger}$ is a projection, $V_{E, \textrm{inv}}^{C'\rightarrow CC'}$ is also an isometry. 
The quantum channel $\cR_E^{C \rightarrow CA}$ is then defined as
\begin{equation}
    \cR_E^{C\rightarrow CA}(\rho^C) := \tr_{C'}\bigl[R_E^{C \rightarrow C C' A} \rho^C R_E^{C \rightarrow C C' A \dagger}\bigr]. \label{Eq:Ekoshietsgh}
\end{equation}

To understand the action of this quantum channel $\cR_E^{C \rightarrow CA}$, we observe that the middle and the right-most terms in~\cref{Eq:cohmeas1} coherently record the measurement outcome of $M_E$ into an ancillary system $A$. The left-most term $V_{E, \textrm{inv}}^{C'\rightarrow CC'}$ is for undoing the action on $C$ as much as possible, so that the back-action to the system is minimized.

We here comment on the fact that the choice of the quantum channel $\cR_E^{C \rightarrow CA}$ is not unique. In fact, we can instead use the isometry 
\begin{equation}
    \sum_{j=0}^{d-1} M_{E,j}^C \otimes \ket{j_E}^A + \sqrt{I^C - \sum_{j=0}^{d-1} (M_{E,j}^C)^2} \otimes \ket{{\rm fail}}^A,    
\end{equation} 
from $C$ to $CA$, where $\ket{{\rm fail}}^A$ is the state in $A$ orthogonal to $\ket{j_E}^A$ for $j=0, \dots, d-1$. Accordingly, $A$ is a $(d+1)$-dimensional system in this case. In~\cref{SS:Thm1Proof}, we explain that the use of this isometry instead of $\cR_E^{C \rightarrow CA}$ leads to the same action.

The second quantum channel $\cQ_F^{CA \rightarrow A}$ is constructed from the other decoding measurement $M_F$ as 
\begin{equation}
    \cQ_F^{CA \rightarrow A}(\rho^{CA})  = \sum_{l=0}^{d-1} \Theta_l^{A} \tr_{C}\bigl[ M_{F,l}^C \rho^{CA} \bigr] \Theta_l^{A \dagger},\label{Eq:defC}
\end{equation}
where
\begin{align}
    &\Theta_l^A :=\sum_{j=0}^{d-1} e^{i {\rm arg} \braket{j_E}{l_F}} \ketbra{j_E}{j_E}^A. \label{Eq:CorrTheta}
\end{align}
This channel corresponds to the operations that the system $C$ is measured by $M_F$, and then, a feed-forward unitary $\Theta_l^{A}$ is applied to $A$ depending on the measurement outcome $l$.
Clearly, this is a generalization of the measurement and feed-forward in the decoding circuit for CSS codes.

Based on these two quantum channels, we define the C-to-Q decoder as follows.
While our primary focus is on the decoding measurements of two CQ codes, the C-to-Q decoder itself can be defined for any pair of two POVMs with $d$ outcomes.

\begin{definition}[Classical-to-Quantum decoder] \label{def:C-to-Q}
    \emph{
   For a given pair of two POVMs with $d$ outcomes, $(M_E,M_F)$, we define a C-to-Q decoder by
    \begin{equation}
        \cD^{C \rightarrow A}_{\rm CtoQ} = \cQ_F^{CA \rightarrow A} \circ \cR_E^{C\rightarrow CA}, \label{Eq:defCtoQ}
    \end{equation}
    where quantum channels $\cR_E^{C\rightarrow CA}$ and $\cQ_F^{CA \rightarrow A}$ are defined by~\cref{Eq:Ekoshietsgh,Eq:defC}, respectively.
    }
\end{definition}

Note that we can construct two C-to-Q decoders from a single pair $(M_E, M_F)$ of POVMs as the two measurements $M_E$ and $M_F$ play different roles in the C-to-Q decoder.
By swapping their roles, two different C-to-Q decoders are obtained, which in general lead to different decoding errors.

The following theorem provides an upper bound on the decoding error by the C-to-Q decoder (see~\cref{SS:Thm1Proof} for the proof).

\begin{restatable}{theorem}{CtoQ}\emph{(Decoding error of the C-to-Q decoder)} \label{Thm:CtoQ}
    Given an encoding channel $\cE^{A \rightarrow B}$ of a QECC for a noisy channel $\cN^{B \rightarrow C}$, let $(E, \cE^{A \rightarrow B}, M_E)$ and $(F, \cE^{A \rightarrow B}, M_F)$ be the associated CQ codes in the $E$- and $F$-bases, respectively. Let $\Delta_{{\rm cl},E}=\Delta_{{\rm cl},E}(M_E|\cN \circ \cE)$ and $\Delta_{{\rm cl},F}=\Delta_{{\rm cl},E}(M_F|\cN \circ \cE)$ be their classical decoding errors.
    The decoding error of the C-to-Q decoder $\cD_{\rm CtoQ}$ constructed from $M_E$ and $M_F$ satisfies
    \begin{multline}
        \Delta_{\rm{q}}(\cD_{\rm CtoQ}|\cN \circ \cE)
        \leq  \sqrt{\Delta_{{\rm cl},E}\bigl(2 - \Delta_{{\rm cl},E} \bigr)}  \\
        + \sqrt{\Delta_{{\rm cl}, F}}
        + \sqrt{\Xi_{EF}}.  
    \end{multline}
    Here, $\Xi_{EF}$ is given by
    \begin{align}
    \Xi_{EF} &:=  1 - \sum_{l=0}^{d-1} \tr \bigl[M_{F, l} \cN \circ \cE(\pi) \bigl] F_{BC}\bigl({\rm unif}_d, p_l\bigr), \\
    &\leq  1 - \min_{l=0,\dots, d-1} F_{BC}\bigl({\rm unif}_d, p_l\bigr), \label{Eq:UB}
    \end{align}
    where $\pi$ is the completely mixed state in $A$, and 
    \begin{equation}
        F_{BC}({\rm unif}_d, p_l)=\biggl(\sum_j \sqrt{\frac{p_l(j)}{d}}\biggr)^2    
    \end{equation}
    is the Bhattacharyya distance between the uniform probability distribution ${\rm unif}_d$ on $[0, d-1]$ and the probability distribution $\{p_l(j)=| \braket{j_E}{l_F} |^2 \}_{j=0}^{d-1}$ determined by the two bases $E$ and $F$.    
\end{restatable}

\cref{Thm:CtoQ} shows that the decoding error of the C-to-Q decoder can be split into three terms: $\Delta_{{\rm cl},E}$, $\Delta_{{\rm cl},F}$, and $\Xi_{EF}$. The first two terms represent classical decoding errors of the associated CQ codes and are analogous to the CSS codes (see~\cref{eq:error_for_css}). 
Thus, \cref{Thm:CtoQ} extends the decoding error of the CSS code decoding circuit to the case of general QECCs.
The last term, $\Xi_{EF}$, quantifies how far the bases $E$ and $F$ are from mutually unbiased bases (MUBs). For instance, $\Xi_{EF} = 0$ if and only if $(E,F)$ is MUBs, and $\Xi_{EF} = 1-1/d$ if and only if $E=F$. 
This term does not appear in~\cref{eq:error_for_css} for CSS codes because CSS codes are constructed from two classical codes for the Pauli-$X$ and -$Z$ bases that are mutually unbiased and $\Xi_{XZ} = 0$. 
The term $\Xi_{EF}$ in~\cref{Thm:CtoQ} is the non-trivial term that quantitatively reveals the fundamental role of complementarity in decoding general QECCs.

Important instances of MUBs are the (generalized) Pauli bases $X$ and $Z$ in multi-qubit systems and the Heisenberg-Weyl group for qudits. In these cases, $\Xi_{EF} = 0$, and the unitaries $\Theta_l^A$ in~\cref{Eq:CorrTheta} has a simple form. For instance, when $E$ and $F$ are set to be $Z$ and $X$, respectively, then $\Theta_l^A=(Z^A)^l$, where the operator $Z^A$ is the (generalized) Pauli-$Z$ operator on $A$. 

Note that $\Xi_{EF}$ can also be bounded from above in the form different from~\cref{Eq:UB} as
\begin{multline}
    \Xi_{EF} \leq 1 - \frac{1}{d} \sum_{l=0}^{d-1}  F_{BC}\bigl({\rm unif}_d, p_l\bigr)\\
    + \Delta_{{\rm cl}, F}\bigl( F_{BC, \max} - F_{BC, \min} \bigr),\label{Eq:koregeri}
\end{multline}
where $F_{BC, \max} = \max_l F_{BC}\bigl({\rm unif}_d, p_l\bigr)$ and $F_{BC, \min} = \min_l F_{BC}\bigl({\rm unif}_d, p_l\bigr)$.
This may be useful especially when $\Delta_{{\rm cl}, F}$ is small.
See~\cref{App:last} for the derivation.

Apart from its fundamental importance,~\cref{Thm:CtoQ} is of practical use for experimentally evaluating the performance of a QECC. In fact,~\cref{Thm:CtoQ} enables us to estimate the QECC performance by estimating classical decoding errors of two CQ codes and by checking the complementarity of the bases that define the classical information. 
Estimating classical decoding errors of CQ codes is practically more tractable than directly quantifying the performance of the QECC: while checking the QECC performance requires handling the system in a fully-quantum manner, classical decoding errors can be evaluated simply by measuring the system.
When we do so, it is important to understand how the complementarity of the bases affects the QEC performance since the presence of experimental imperfection, such as calibration errors in the device, may unintentionally make the encoding bases deviate from MUBs. 
By evaluating such errors in the device and by combining it with the classical decoding errors, the QECC performance can be estimated by~\cref{Thm:CtoQ}. 

Another possible use of the C-to-Q decoder is for switching code by preparing the ancillary system in the QECC different from the QECC for the noisy channel. This may be of particular importance in quantum communication since the QECC in a local environment shall have different properties from the QECC for transmitting quantum information: fault-tolerance would be the most important in the former, and the QECC with higher encoding rate is more preferable in the latter. 
By using the C-to-Q decoder, one can simultaneously correct errors and switch QECCs.\\

Before we move on, we provide a high-level explanation about how the C-to-Q decoder works, which is basically for the same reason as the decoding circuit of the CSS code as explained in~\cref{SSS:DecodeCSS}.
First, because the quantum channel $\cR_E^{C \rightarrow CA}$ is designed to coherently record the measurement outcome of $M_E$ into the ancillary system $A$, it generates a coherent correlation between the noisy system $C$ and the ancillary system $A$ in the basis $E$.
In fact, when the POVM $M_E$ satisfies $\Delta_{{\rm cl},E}(M_E|\cN \circ \cE) = 0$, it holds that (see~\cref{S:CtoQProof})
\begin{multline}
    \cR_E^{C \rightarrow CA}\circ \cN^{B \rightarrow C} \circ \cE^{A \rightarrow B}(\Phi^{AR})\\
    = \frac{1}{d}\sum_{j,i}\ketbra{j_E^*}{i_E^*}^R \otimes \cN^{B \rightarrow C} \circ \cE^{A \rightarrow B}(\ketbra{j_E}{i_E}^A) \otimes  \ketbra{j_E}{i_E}^A. \label{Eq:nksdo34p}
\end{multline}
Clearly, this state is a noisy GHZ state over $R$, $C$, and $A$, where we mean by noisy that $\cN \circ \cE$ is applied onto $A$.
With this state, the remaining task is to transform the noisy GHZ state to the maximally entangled state between $R$ and $A$.

This can be accomplished by the measurement and feed-forward. For instance, a simple GHZ state $\ket{000}_Z + \ket{111}_Z$ can be transformed to a maximally entangled state $\ket{00}_Z + \ket{11}_Z$ by measuring one qubit in the $X$ basis and by applying a feed-forward unitary, which is simply $Z^m$ with $m=0,1$ being the measurement outcome. 
This mechanism is known as the quantum eraser since the key in this process is to completely erase the $Z$-information from one qubit by the measurement and feed-forward. In the case of the C-to-Q decoder, the noisy GHZ state (\cref{Eq:nksdo34p}) can be transformed to the maximally entangled state if one could measure the noisy system $C$ in the basis complementary to $E$ and apply a feed-forward unitary to $A$.
However, we know only the error when $C$ is measured either by $M_E$ or $M_F$, forcing us to apply $M_F$ onto $C$ instead of the ideal one. This results in an additional error quantified by $\Xi(E,F)$.\\

An immediate corollary of~\cref{Thm:CtoQ} is the following. 

\begin{corollary} \label{Cor:CtoQMUB}
Under the same setting as in~\cref{Thm:CtoQ}, let $(E, F)$ be a pair of MUBs. Then, the C-to-Q decoder $\cD_{\rm CtoQ}$ satisfies
\begin{align}
    \Delta_{\rm{q}}(\cD_{\rm CtoQ}|\cN \circ \cE) 
    &\leq 
    \sqrt{\Delta_{{\rm cl},E}\bigl(2 - \Delta_{{\rm cl},E} \bigr)}
    + \sqrt{\Delta_{{\rm cl}, F}}\label{Eq:DNZ} \\
    &\leq 
    (1 + \sqrt{2}) \max_{W= E, F}\sqrt{\Delta_{{\rm cl},W}}.
\end{align}
\end{corollary}

This corollary improves by constant factors the decoding errors of the previous decoding strategies~\cite{K2007, BR2009, T2010, Renes2016}, which are all different.

The decoding circuit for CSS codes, shown in~\cref{Fig:CSSCodes}, is an important instance to which~\cref{Cor:CtoQMUB} can be applied. However, one may notice that the upper bound on the decoding error by the decoding circuit for CSS codes (\cref{eq:error_for_css}) is better than~\cref{Eq:DNZ}. 
The difference comes from the fact that the error in the quantum channel $\cQ_F^{CA \rightarrow C}$ for the measurement and feed-forward may in general be affected by the preceding procedure $\cR_E^{C \rightarrow CC'A}$. In contrast, the two decoding procedures for CSS codes are independent and do not affect each other since the logical-$Z$ and -$X$ information are orthogonally encoded in the case of CSS codes. See~\cref{App:CSS} for the further details.

\subsubsection{Implications of the C-to-Q decoder} \label{SSS:b}
\cref{Cor:CtoQMUB} has two immediate implications.
The first implication is that, when decoding measurements of the CQ codes are suitably chosen, the C-to-Q decoder is nearly optimal. 

To see this, let $\cD_{opt}$ be an optimal decoder for an encoder $\cE$ and a noisy map $\cN$, that is, $\Delta_{\rm{q}}(\cD_{opt}|\cN \circ \cE) \leq \Delta_{\rm{q}}(\cD|\cN \circ \cE)$ for any decoder $\cD$. From the optimal decoder $\cD_{opt}$, we can trivially construct two decoding measurements $M_W$ ($W=E, F$) for the CQ codes that satisfy $\Delta_{{\rm cl}, W}(M_W|\cN \circ \cE) \leq \Delta_{\rm{q}}(\cD_{opt}|\cN \circ \cE) $. The C-to-Q decoder based on such measurements $M_E$ and $M_F$ satisfies
\begin{multline}
    \Delta_{\rm{q}}(\cD_{opt}|\cN \circ \cE) \leq \Delta_{\rm{q}}(\cD_{\rm CtoQ}|\cN \circ \cE) \\
    \leq  (1 + \sqrt{2}) \sqrt{\Delta_{\rm{q}}(\cD_{opt}|\cN \circ \cE)}, \label{Eq:trivial}
\end{multline}
where we have used~\cref{Cor:CtoQMUB} to obtain the last inequality. Thus, the decoding error by the C-to-Q decoder can be no worse than a square root of the optimal one. 
This immediately implies that the C-to-Q decoder combined with a good encoder is capable to achieve the quantum capacity, both entanglement non-assisted~\cite{L1997, S2002, D2005} and assisted ones~\cite{BSST1999,BSST2002}. In this sense, the C-to-Q decoder is a capacity-achieving decoder.

The second implication is about encoding operations that achieve capacities of noisy quantum channels. 
For a given noisy quantum channel $\cN^{B \rightarrow C}$, the largest amount of quantum information that can be reliably transmitted by the channel is called the quantum capacity $C_q(\cN)$~\cite{L1997,S2002,D2005}. 
The quantum capacity is achieved by suitable QECCs. 
Similarly, the largest amount of classical information reliably transmittable by the noisy channel is referred to as the classical capacity $C_c(\cN)$~\cite{SW1998,H1998}, which is achieved by suitable CQ codes.
It trivially holds that $C_q(\cN) \leq C_c(\cN)$ for any noisy channel $\cN$ as classical information is easier to transmit than quantum one.

From~\cref{Cor:CtoQMUB}, we observe that, unless $C_q(\cN) = C_c (\cN)$, an encoder for the CQ code in the $X$-basis that achieves the classical capacity cannot serve as a good encoder for the CQ code in the $Z$-basis.
This conclusion follows by contradiction: let $\cE^{A \rightarrow B}_c$ be the encoder for the CQ code in the $X$-basis that achieves the classical capacity. If $\cE_{c}$ were also a good encoder for the CQ code in the $Z$-basis that achieves the classical capacity, then~\cref{Cor:CtoQMUB} would imply that $C_q(\cN) = C_c (\cN)$.
Therefore, unless $C_q(\cN) = C_c (\cN)$, a good encoder for a CQ code must be specialized to a specific basis.
This is of theoretical interest, as it suggests that the complementarity principle imposes a limitation on the optimal encoding operations for CQ codes.

\subsection{Decoding classical information by projection-based PGMs} \label{SS:pPGMs}
Using~\cref{Thm:CtoQ}, the task of constructing a decoder for a general QECC can be reduced to constructing decoding measurements for the associated CQ codes. To complete the analysis, we focus on the latter problem and provide a sufficient condition for decoding a CQ code.
In this context, as we are concerned only with classical information, we do not explicitly specify the basis and use simplified notation such as $\ket{j}$.

A commonly used decoder for a CQ code is a PGM that is known to achieve the classical capacity.~\cite{SW1998, H1998}.  Here, we slightly modify the PGM and provide an upper bound on the classical decoding error.
We call the modified one projection-based PGM (pPGM), which is basically the same idea as that in~\cite{H1998} and is sometimes referred to PGM as well in the literature.
Let $\Pi_{j}^{C}$ be a projection onto the support of $\cN^{B \rightarrow C} \circ \cE^{A \rightarrow B} (\ketbra{j}{j}^A)$. We define a POVM as
\begin{equation}
    M_{{\rm pPGM}} = \bigl\{ \Pi^{-1/2} \Pi_{j} \Pi^{-1/2} \bigr\}_j,
\end{equation}
where $\Pi := \sum_j \Pi_{ j}$. 

Following the conventional analysis, a sufficient condition for decoding classical information by the pPGM can be obtained. See~\cref{S:pPGMs} for the proof.

\begin{restatable}{proposition}{SC} \label{Prop:SC}
    Let $\cE^{A \rightarrow B}$ be an encoding channel of a QECC for a noisy channel $\cN^{B\rightarrow C}$, and $\tau_{\pi}^C$ and $\tau_j^C$ ($ j = 0, \dots, d-1$) be defined as
    \begin{multline}
        \tau_{\pi}^C:= \cN^{B \rightarrow C} \circ \cE^{A \rightarrow B}(\pi^A), \\
        \text{\ \ and\ \ } \tau_j^C:=\cN^{B \rightarrow C} \circ \cE^{A \rightarrow B}(\ketbra{j}{j}^A),
    \end{multline}
    where $\pi^A$ is the completely mixed state and $W=\{\ket{j}\}_{j=0}^{d-1}$ is a basis. The classical decoding error of the pPGM $M_{{\rm pPGM}}$ satisfies
    \begin{align}
        \Delta_{{\rm cl}, W}(M_{{\rm pPGM}}|\cN \circ \cE) 
        &\leq
        \frac{1}{d \lambda_{\rm min}} \sum_{i \neq j} \tr \bigl[ \tau_i \tau_j \bigr],
    \end{align}
    which can be further reduced to 
    \begin{multline}
        \Delta_{{\rm cl}, W}(M_{{\rm pPGM}}|\cN \circ \cE) \\
        \leq \frac{1}{\lambda_{{\rm min}}} \biggl\{ \frac{d}{2^{  H_2(C)_{\tau_{\pi}}}} - \frac{1}{d} \sum_{j=0}^{d-1} \frac{1}{2^{H_2(C)_{\tau_j}}} \biggr\}. \label{Eq:nokafernvroe;rkare}
    \end{multline}
    Here, $\lambda_{{\rm min}}:= \min_{j \in [0, d-1]} \lambda_{\rm min} ( \tau^C_{j})$ with $\lambda_{\rm min}(\sigma)$ being the minimum non-zero eigenvalue of $\sigma$, and $H_2(C)_{\rho}$ is the collision entropy of $\rho$.
\end{restatable}

\begin{figure*}[tb!]
\centering
\includegraphics[width=0.8\textwidth,clip]{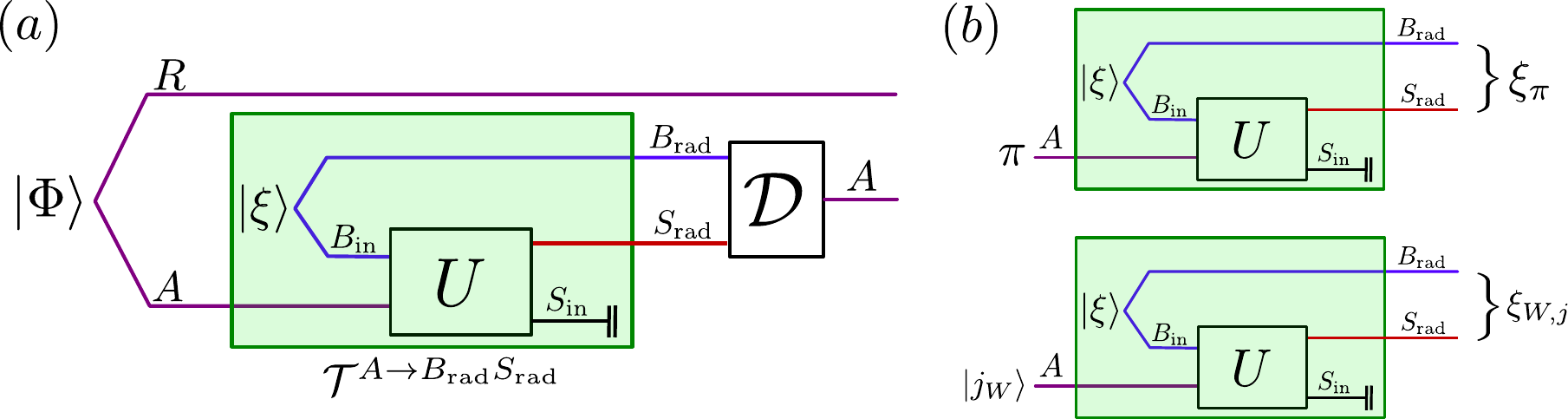}
\caption{(a) Diagram of the Hayden-Preskill protocol. The goal is to find the best possible decoder $\cD^{B_{\rm rad} S_{\rm rad} \rightarrow A}$ such that $\rho_{\rm fin}^{AR} \approx \Phi^{AR}$. See the main text for the details. (b) Definitions of $\xi_{W,j}$ and $\xi_{\pi}$ defined in~\cref{Eq:34vlp'pmer,Eq:34vlp'pmer2}.}
\label{Fig:decoderHP}
\end{figure*}

\cref{Prop:SC} implies that the classical decoding error of a CQ code by the pPGM is characterized by the minimum non-zero eigenvalue of $\tau_j$, and the collision entropies of $\tau_{\pi}$ and $\tau_j$.
Since the collision entropy of $\tau_j$ is closely related to the minimum non-zero eigenvalues, $\lambda_{{\rm min}}$ and $H_2(C)_{\tau_{\pi}}$ would be the most important quantities.
When the minimum non-zero eigenvalue is close to zero, we may exploit the technique of \emph{smoothing} (see, e.g.,~\cite{RS2004}), but we leave the explicit use of smoothing as a future problem.

\cref{Prop:SC} can be regarded as a natural sufficient condition for successful decoding of the CQ code.
We observe from~\cref{Prop:SC} that the classical decoding error is small if $d2^{-H_2(C)_{\tau_{\pi}}} \approx \frac{1}{d}\sum_j 2^{-H_2(C)_{\tau_{j}}}$. This is achieved when there is no collision by $\cN \circ \cE$ in the sense that each pair $\ket{j}$ and $\ket{i}$ $(i \neq j)$ is mapped by the channel $\cN \circ \cE$ to a pair of states with negligible overlaps in their supports. Furthermore,~\cref{Eq:nokafernvroe;rkare} shows that each overlap should be negligibly small compared to $\lambda_{\rm min}$.
When these conditions are met, it is naturally expected that the pPGM works well as a decoder for classical information.

\subsection{Decoding Hayden-Preskill by the C-to-Q decoder} \label{SS:HPCtoQ}

To demonstrate the power of the C-to-Q decoder, we apply it to the Hayden-Preskill protocol~\cite{HP2007}. The protocol, originally proposed in the context of the black hole information paradox, is closely related to random coding under erasure noise. 
We begin with an overview of the protocol in~\cref{SSS:HP} and informally present our results in~\cref{SSS:HPresults}.

\subsubsection{Setting of the Hayden-Preskill protocol} \label{SSS:HP}

The Hayden-Preskill protocol is a qubit toy model of the black hole information paradox~\cite{HP2007} and has been extensively studied from various perspectives~\cite{HQRY2016,RY2017,NWK2020,TTK2022}.
Let $B_{\rm in}$ be an $N$-qubit system representing a quantum black hole, with an initial state $\xi^{B_{\rm in}}$.
We denote by $\ket{\xi}^{B_{\rm in} B_{\rm rad}}$ a purification of $\xi^{B_{\rm in}}$ by a purifying system $B_{\rm rad}$, which corresponds to the \emph{past} Hawking radiation.
A quantum information source $A$ of $k$ qubits, defined by a state $\Phi^{AR}$ maximally entangled with a reference system $R$, is thrown into the black hole, enlarging the BH $B_{\rm in}$ to $S:=A B_{\rm in}$ of $N+k$ qubits. The composite system $S$ then undergoes unitary dynamics $U^S$ and is split into two subsystems: $S_{\rm in}$ with $N+k-\ell$ qubits and $S_{\rm rad}$ with $\ell$ qubits. Note that $S= AB_{\rm in}=S_{\rm in} S_{\rm rad}$.
The subsystem $S_{\rm rad}$ represents the \emph{new} Hawking radiation. 
The entire protocol is illustrated in~\cref{Fig:decoderHP} (a). 

The task of the protocol is to decode the $k$-qubit information of $A$ from the past and new radiations $B_{\rm rad} S_{\rm rad}$. The decoding error, when decoded by $\cD^{B_{\rm rad} S_{\rm rad} \rightarrow A}$, is denoted by
\begin{equation}
\Delta_{\rm{q}}(\cD | \xi, U)
=
\frac{1}{2}\bigl \|\Phi^{AR} - \cD^{B_{\rm rad} S_{\rm rad} \rightarrow A}(\rho^{B_{\rm rad} S_{\rm rad} R}_{\rm fin} )\bigr\|_1, \label{Eq:Delta19}
\end{equation}
where $\rho^{B_{\rm rad}S_{\rm rad}R}_{\rm fin} := \tr_{S_{\rm in}} \bigl[ \cU^S(\Phi^{AR} \otimes \xi^{B_{\rm in} B_{\rm rad}}) \bigr]$, and $\cU(\rho) = U \rho U^{\dagger}$.

As explained in~\cite{HP2007}, the Hayden-Preskill protocol is analogous to transmitting quantum information through a quantum erasure channel~\cite{BDS1997,BSST1999} with an erasure rate $1-\frac{\ell}{N+k}$.
However, unlike typical settings in quantum information theory, the encoder in the Hayden-Preskill protocol is fixed to the specific unitary dynamics $U^S$.

In~\cite{HP2007}, an upper bound on the average decoding error was derived under the assumption that $U^S$ is a Haar random unitary. 

\begin{theorem}[Previous result about the Hayden-Preskill protocol~\cite{HP2007, DBWR2010}] \label{Thm:HP}
    When the unitary $U^S$ is Haar random, there exists a family of decoders $\cD$ whose decoding error satisfies
    \begin{equation}
        \EH \bigl[ \inf_{\cD} \Delta_{\rm{q}}(\cD | \xi, U) \bigr] \leq 2^{\frac{1}{2}(\ell_{\rm th} -\ell)},
    \end{equation}
    where $\ell_{\rm th} := k+ \frac{N-H_2(B_{\rm in})_{\xi}}{2}$. 
\end{theorem}
This shows that the optimal decoding error on average is characterized by the collision entropy $H_2(B_{\rm in})_{\xi}$ of the initial state $\xi^{B_{\rm in}}$ of $B_{\rm in}$. As the entropy ranges from $0$ to $N$, the threshold $\ell_{\rm th}$ varies from $k$ to $k + N/2$. 
However, a decoder to achieve the bound was not explicitly provided since the analysis of~\cref{Thm:HP} was based on the decoupling approach. 
Later, two decoders were proposed. One is based on the Petz recovery map~\cite{PSSY2022}, and the other is in the form of quantum circuits but works only in a special case~\cite{YK2017}. The latter one was further extended to a more general one~\cite{UN2024}.

Note that~\cref{Thm:HP} also implies that, for any basis $W$, there exists a decoding measurement $M_W$ for classical information that leads to the decoding error such as
\begin{multline}
    \mbb{E}_{U \sim {\sf H}} \bigl[ \Delta_{{\rm cl}, W}(M_W|\xi, U) \bigr]
   \leq  \\
    \mbb{E}_{U \sim {\sf H}} \bigl[ \Delta_{\rm{q}}(\cD|\xi, U) \bigr] \leq 2^{(\ell_{\mathrm{th}} - \ell )/2}. \label{Eq:HPclassicalprevious}
\end{multline}
This is merely because classical errors are always smaller than quantum errors.

\subsubsection{Decoding the Hayden-Preskill protocol by the C-to-Q decoder with pPGMs} \label{SSS:HPresults}

We now apply the C-to-Q decoder constructed from pPGMs to the Hayden-Preskill protocol and investigate decoding errors for classical and quantum information.
This can be done by formulating the Hayden-Preskill protocol using the following quantum channel:
\begin{equation}
    \cT^{A \rightarrow B_{\rm rad} S_{\rm rad}}(\rho^A) = \tr_{S_{\rm in}}\bigl[ U^S (\rho^A \otimes \xi^{B_{\rm in} B_{\rm rad}}) U^{\dagger S} \bigr], \label{Eq:TinHP}
\end{equation}
where $S = A B_{\rm in} = S_{\rm in} S_{\rm rad}$ (see~\cref{Fig:decoderHP} (a) as well), and $d = \dim \cH^A=2^k$. 
This map can be decomposed into an ``encoding'' map $\cE^{A \rightarrow B_{\rm rad}S}$ and a ``noisy''map $\cN^{S \rightarrow S_{\rm rad}}$, such as $\cT^{A \rightarrow B_{\rm rad} S_{\rm rad}} = \cN^{S \rightarrow S_{\rm rad}} \circ \cE^{A \rightarrow B_{\rm rad}S}$, where
\begin{align}
    &\cE^{A \rightarrow B_{\rm rad}S}: \rho^A \mapsto U^S (\rho^A \otimes \xi^{B_{\rm in} B_{\rm rad}}) U^{\dagger S},\\
    &\cN^{S \rightarrow S_{\rm rad}} = \tr_{S_{\rm in}}.
\end{align}
We then consider two CQ codes $(X, \cE^{A \rightarrow B_{\rm rad}S}, M_X^{B_{\rm rad} S_{\rm rad}})$ and $(Z, \cE^{A \rightarrow B_{\rm rad}S}, M_Z^{B_{\rm rad} S_{\rm rad}})$: the former in the Pauli-$X$ basis and the latter in the Pauli-$Z$ basis.
Their classical decoding errors are denoted by $\Delta_{{\rm cl}, X}(M_X |\xi, U)$ and $\Delta_{{\rm cl}, Z}(M_Z |\xi, U)$, respectively. 

In the Hayden-Preskill protocol, the dynamics $U$ is commonly assumed to be Haar random except a few cases~\cite{NT2024}. Before we rely on this assumption, we explicitly write down the condition in~\cref{Prop:SC}.
To this end, we use the following states: for $W = X, Z$,
\begin{align}
&\xi_{W, j}^{B_{\rm rad}S_{\rm rad}}= \cT^{A \rightarrow B_{\rm rad} S_{\rm rad}}(\ketbra{j_W}{j_W}^A ), \label{Eq:34vlp'pmer}\\
&\xi_{\pi}^{B_{\rm rad}S_{\rm rad}}= \cT^{A \rightarrow B_{\rm rad} S_{\rm rad}}(\pi^A ).\label{Eq:34vlp'pmer2}
\end{align}
See~\cref{Fig:decoderHP} (b) for the diagrams of these states.
Then, a sufficient condition for decoding $W$-classical information by the corresponding pPGM $M_{{\rm pPGM}, W}$ is rewritten as
\begin{multline}
    \Delta_{{\rm cl},W}(M_{{\rm pPGM},W}|\cN \circ \cE) 
    \leq \\ \frac{1}{\lambda_{W, {\rm min}}} \biggl( 2^{k - H_2(B_{\rm rad}S_{\rm rad})_{\xi_{\pi}}}
    - \frac{1}{2^k} \sum_j 2^{- H_2(B_{\rm rad}S_{\rm rad})_{\xi_{W, j}}} \biggr).
\end{multline}
By combining this with~\cref{Thm:CtoQ} about the C-to-Q decoder, we obtain a criteria for decoding quantum information.
As the collision entropy has gravity interpretation~\cite{D2016}, the criteria may provide a new decoding criteria in terms of gravity and space-time geometry.

We now assume that $U^S$ is a Haar random unitary. In this case, we obtain the following result. See~\cref{SS:FHP} for the formal statement.

\begin{theorem}[Informal statement] \label{Thm:HPformal}
    Suppose that ${\rm rank}(\xi^{B_{\rm in}}) \lambda_{\rm min}(\xi^{B_{\rm in}}) <1/2$ and that
    \begin{equation}
        (N+k-\ell)2^{- (k + 2(\ell - \ell_{\mathrm{th}} ))} \approx 0, \label{Eq:infty1}
    \end{equation}
    for sufficiently large $N$. Then, the classical decoding error by the pPGM satisfies
    \begin{align}
        &\EH \bigl[ \Delta_{{\rm cl}, W}(M_{{\rm pPGM}}|\xi, U) \bigr]
        \lesssim 
        2^{2(\ell_{\rm th}-\ell)} , \label{Eq:coio34/p1}
    \end{align}
    for any basis $W$. The decoding error by the C-to-Q decoder satisfies
    \begin{align}
        &\EH \bigl[ \Delta_{\rm{q}}(\cD_{\rm CtoQ}|\xi, U) \bigr] 
        \lesssim 
        (1 + \sqrt{2}) 2^{\ell_{\rm th} -\ell}, \label{Eq:coio34/p}
    \end{align}
    where $\ell_{\rm th} = k+ \frac{N-H_2(B_{\rm in})_{\xi}}{2}$. 
\end{theorem}

This result improves the decoding errors in the Hayden-Preskill protocol compared to the previous result based on the decoupling approach, provided in~\cref{Thm:HP}. 
About decoding \emph{quantum} information, the improvement is by factor 2 in the error exponent, while the threshold $\ell_{\rm th}$ is the same as previous. The latter should be the case since the threshold corresponds to the quantum capacity of the Hayden-Preskill protocol and cannot be improved.

On the other hand, regarding \emph{classical} information, the improvement on the error exponent is by factor $4$ (see~\cref{Eq:HPclassicalprevious} for the previous bound). Although this is just an upper bound, the following proposition follows due to the nature of the C-to-Q decoder.

\begin{proposition} \label{Prop:HPclassicalOptimal}
    In the Hayden-Preskill protocol, the threshold value of $\ell$ for classical information to be decodable from the past and new radiations, $B_{\rm rad}$ and $S_{\rm rad}$ is $\ell_{\rm th} = k + \frac{N- H_2(B_{\rm rad})_{\xi}}{2}$.
\end{proposition}

\begin{Proof}
    Suppose that classical information is decodable from qubits fewer than $\ell_{\rm th}$. We can then construct a C-to-Q decoder that is able to decode quantum information from the same number of qubits, but this is not possible since $\ell_{\rm th}$ is the threshold for decoding quantum information. Hence, $\ell_{\rm th}$ is also the threshold for classical information. $\hfill \qed$
\end{Proof}

\begin{figure*}[tb!]
\centering
\includegraphics[width=.8\textwidth,clip]{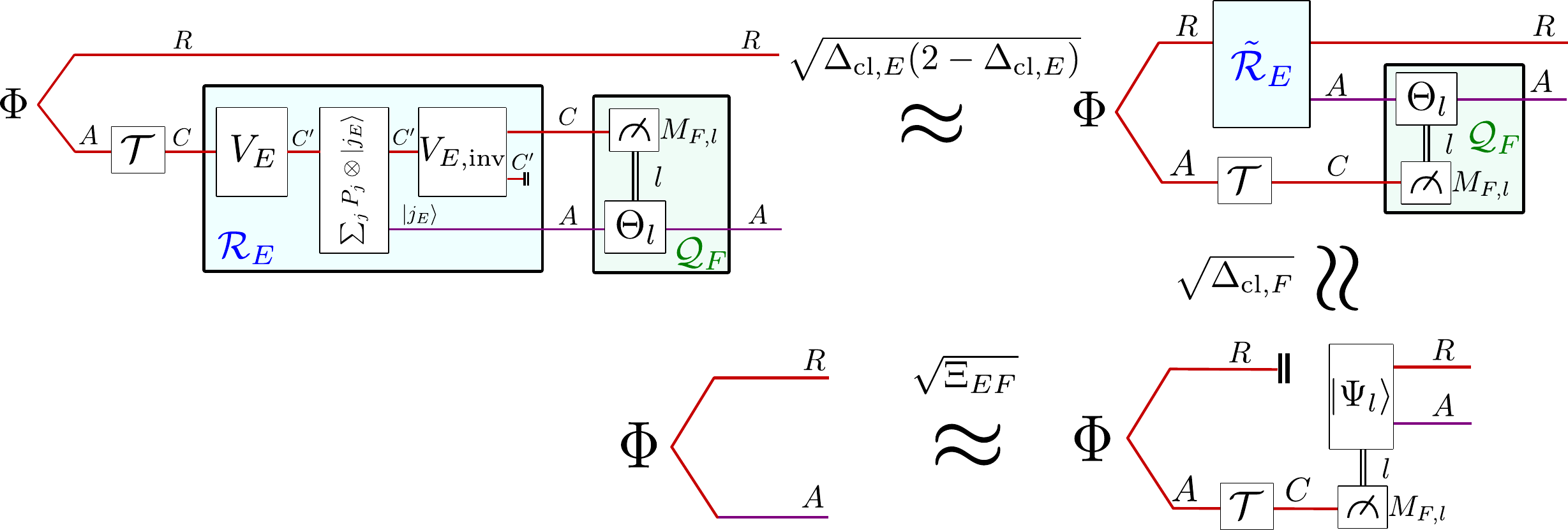}
\caption{A sketch of the proof of~\cref{Thm:CtoQ}. The approximation are in terms of the trace distance. For the first, second, and last ones, see~\cref{Eq:Xerror,Eq:asjif;jew42t,Eq:korefeko}, respectively.
}
\label{Fig:Proof}
\end{figure*}

\cref{Prop:HPclassicalOptimal} implies that the threshold value of $\ell$ for the classical information to be decodable is exactly the same as that for the quantum information.
The difference between the recovery of classical and quantum information in the Hayden-Preskill protocol is only in terms of the error exponent: the recovery error of classical information decreases as $\ell$ increases quadratically faster than that of quantum information.
This may be of surprise to some extent since classical information is in general easier to recover than quantum information.

In fact, these results lead to an interesting observation that a quantum black hole with Haar random dynamics is an optimal encoder of quantum information but is a poor encoder of classical information. This is best illustrated when the initial state $\xi$ of the initial black hole $B_{\rm in}$ is in a pure state, in which case $\ell_{\rm th}= k + N/2$.
While this number of qubits is necessary and sufficient for decoding classical information when it is encoded by the black hole dynamics, there exist better encoding schemes by which $k$-bit classical information are decodable merely from a little more than $k$ qubits. Thus, the black hole dynamics is too random to optimally encode classical information. 

We finally comment on the two assumptions on~\cref{Thm:HPformal}.
The first one, ${\rm rank}(\xi^{B_{\rm in}}) \lambda_{\rm min}(\xi^{B_{\rm in}}) <1/2$, is about the property of the state $\xi^{B_{\rm in}}$. This holds especially when all the non-zero eigenvalues of $\xi^{B_{\rm in}}$ are nearly equal, which is the case for any pure state or the completely mixed state. 
In contrast, the second one puts a mild but additional assumption that was not taken in previous studies. For instance, the second assumption is satisfied when $k>\log N$, or $k$ is constant and $\ell \geq \ell_{\rm th} + \log N$.
We strongly believe that these two assumption are just for a technical reason and should be removable.

\section{Proofs of Theorem~\ref{Thm:CtoQ} and Proposition~\ref{Prop:SC}} \label{S:CtoQProof}

\subsection{Proof of~\cref{Thm:CtoQ}} \label{SS:Thm1Proof}

We now prove~\cref{Thm:CtoQ}. 

\CtoQ*

In the proof, we denote the composed channel of $\cE^{A \rightarrow B}$ and $\cN^{B \rightarrow C}$ by $\cT^{A \rightarrow B}$. That is,
\begin{equation}
    \cT^{A \rightarrow B} := \cN^{B \rightarrow C} \circ \cE^{A \rightarrow B}.
\end{equation}
We also use the notation $\Phi^{CR}_{\cT} := \cT^{A \rightarrow C}(\Phi^{AR})$, where $\Phi^{AR}$ is the maximally entangled state between $A$ and $R$.

\begin{proof}[Proof of~\cref{Thm:CtoQ}]

As given in~\cref{Eq:defCtoQ}, the C-to-Q decoder consists of two quantum channels $\cR_E^{C\rightarrow CA}$ and $\cQ_F^{CA \rightarrow A}$ as
\begin{equation}
    \cD^{C \rightarrow A}_{\rm CtoQ} = \cQ_F^{CA \rightarrow A} \circ \cR_E^{C\rightarrow CA}.    
\end{equation}
We show that, when decoding errors for the CQ codes are small and the bases $E$ and $F$ are nearly mutually unbiased, these channels can be replaced with different ones acting on different systems, which eventually leads to~\cref{Thm:CtoQ}. See also~\cref{Fig:Proof}.

When $\Delta_{{\rm cl}, E}$ is small, it holds that
\begin{align}
    \cR_E^{C\rightarrow CA}(\Phi_{\cT}^{CR})
    &\approx
    \tilde{\cR}_E^{R \rightarrow R A}(\Phi_{\cT}^{CR}), \label{Eq:korer}
\end{align}
where $\tilde{R}_E^{R\rightarrow RA}  := \sum_{j} \ketbra{j_E^*}{j_E^*}^R \otimes \ket{j_E}^A$. 
To see this, notice from~\cref{Eq:cohmeas1,Eq:unitarize} that, since $\ket{e_0}^{C'}$ lies in the range of the isometry $V_E^{C\rightarrow C'}$, we have
\begin{align}
    \bra{e_0}^{C'} R_E^{C\rightarrow CC'A} &= \sum_{j=0}^{d-1} V_E^{C\rightarrow C' \dagger}P_j^{C'}V_E^{C\rightarrow C'}\otimes \ket{j_E}^A \\
    &= \sum_{j=0}^{d-1} M_{E,j}^{C}\otimes \ket{j_E}^A, \label{Eq:Odef}
\end{align}
which leads to
\begin{align}
&R_E^{C\rightarrow CC'A \dagger}\bigl(\tilde{R}_E^{R \rightarrow RA} \otimes \ket{e_0}^{C'}\bigr) \nonumber \\
&=  \biggl(  \sum_{j} M_{E,j}^{C} \otimes \bra{j_E}^{A} \biggr)\biggl(\sum_{i} \ketbra{i_E^*}{i_E^*}^R \otimes \ket{i_E}^A\biggr)\\
&= \sum_{j} M_{E,j}^{C} \otimes \ketbra{j_E^*}{j_E^*}^R.  \label{Eq:jyusyo}
\end{align}

We also mention that we can use another isometry $R_E^{'C \rightarrow CA}= \sum_{j=0}^{d-1} M_{E,j}^C \otimes \ket{j_E}^A + \sqrt{I^C - \sum_{j=0}^{d-1} (M_{E,j}^C)^2} \otimes \ket{{\rm fail}}^A$ instead of $R_E^{C\rightarrow CC'A}$. In fact, it holds that 
\begin{equation}
    R_E^{'C \rightarrow CA \dagger} \tilde{R}_E^{R \rightarrow RA} = \sum_{j} M_{E,j}^{C} \otimes \ketbra{j_E^*}{j_E^*}^R,
\end{equation}

Using a purification $\ket{\Phi_{\cT}}^{CRS}$ of $\Phi_{\cT}^{CR}$ by the system $S$, we obtain
\begin{widetext}
\begin{align}
    \bigl\| \cR_E^{C \rightarrow CA}(\Phi_{\cT}^{CR}) -  \tilde{\cR}_E^{R\rightarrow RA}(\Phi_{\cT}^{CR}) \bigr\|_1
    &\leq 
    \bigl\| R_E^{C \rightarrow CC'A}\Phi_{\cT}^{CRS} \bigl(R_E^{C \rightarrow CC'A }\bigr)^{\dagger} -  \tilde{\cR}_E^{R\rightarrow RA}(\Phi_{\cT}^{CRS}) \otimes \ketbra{e_0}{e_0}^{C'} \bigr\|_1\\
    &=
    2\sqrt{1 - \bigl|  \bra{\Phi_{\cT}}^{CRS} \bigl( R_E^{C \rightarrow CC'A } \bigr)^{\dagger} \tilde{R}_E^{R \rightarrow RA}  (\ket{\Phi_{\cT}}^{CRS} \otimes \ket{e_0}^{C'} )\bigr|^2}\\
    &=
    2\sqrt{1 - \bigl|  \bra{\Phi_{\cT}}^{CRS} \bigl(\sum_{j} M_{E,j}^{C} \otimes \ketbra{j_E^*}{j_E^*}^R \bigr)  \ket{\Phi_{\cT}}^{CRS} \bigr|^2},\\
    &=
    2\sqrt{1 -  \Bigl(\tr \bigl[ \Phi_{\cT}^{CR} \bigl(\sum_{j} M_{E,j}^{C} \otimes \ketbra{j_E^*}{j_E^*}^R \bigr)\bigr]\Bigr)^2},
\end{align}
\end{widetext}
where the first line follows from the monotonicity of the trace norm, the second from the fact that both states are pure, and the third from~\cref{Eq:jyusyo}.

As $\bra{j_E^*}^R\Phi_{\cT}^{CR} \ket{j_E^*}^R = d^{-1}\cT^{A \rightarrow C}(\ketbra{j_E}{j_E}^A)$, we have
\begin{multline}
    \tr \bigl[ \Phi_{\cT}^{CR} \bigl(\sum_{j} M_{E,j}^{C} \otimes \ketbra{j_E^*}{j_E^*}^R \bigr)\bigr] \\
    =
    \frac{1}{d} \sum_j \tr \bigl[\cT^{A \rightarrow C}(\ketbra{j_E}{j_E}^A) M_{E,j}^{C} \bigr] = 1- \Delta_{{\rm cl},E}.
\end{multline}
Thus, we arrive at
\begin{multline}
    \bigl\| \cR_E^{C \rightarrow CA}(\Phi_{\cT}^{CR}) -  \tilde{\cR}_E^{R\rightarrow RA}(\Phi_{\cT}^{CR}) \bigr\|_1 \\
    \leq 
    2 \sqrt{\Delta_{{\rm cl},E} ( 2 - \Delta_{{\rm cl},E})}. \label{Eq:Xerror}
\end{multline}

We next show that $\cQ_F^{CA \rightarrow A} \circ \tilde{\cR}_E^{R\rightarrow RA}(\Phi_{\cT}^{CR}) \approx \Psi^{AR}$ if $\Delta_{{\rm cl}, F}$ is small, where
\begin{equation}
    \Psi^{AR} := \sum_l q(l)\ketbra{\Psi_l}{\Psi_l}^{AR}, \label{Eq:defPsi}
\end{equation}
with $q(l) = \tr\bigl[\cT^{A\rightarrow C}(\pi^A) M_{F,l}^C \bigr]$
and
$\ket{\Psi_l}^{AR} := \sum_j |\braket{j_E}{l_F}| \ket{j_E}^A \otimes \ket{j_E^*}^R$.
To this end, we decompose the quantum channel $\cQ_F^{CA \rightarrow A}$ as $\cQ_F^{CA \rightarrow A} = \sum_l \cQ_{F, l}^{CA \rightarrow A}$, where $\cQ_{F, l}^{CA \rightarrow A}$ is trace non-increasing CP map given by
\begin{equation}
    \cQ_{F, l}^{CA \rightarrow A}(\rho^{CA}):= \Theta_l^A \tr_C\bigl[\rho^{CA} M_{F, l}^C \bigr] \Theta_l^{A \dagger}.
\end{equation}
Note that 
\begin{align}
    \tr \bigl[\cQ_{F,l}^{CA \rightarrow A} \circ \tilde{\cR}_E^{R\rightarrow RA}(\Phi_{\cT}^{CR}) \bigr]
    &= \tr \bigl[M_{F, l}^C \tilde{\cR}_E^{R\rightarrow RA}(\Phi_{\cT}^{CR}) \bigr]\\
    &=\tr \bigl[M_{F, l}^C \Phi_{\cT}^{CR} \bigr]\\
    &=\tr \bigl[M_{F, l}^C \cT^{A \rightarrow C}(\pi^A) \bigr]\\
    &=q(l),
\end{align}
where the first line follows from the unitarity of $\Theta_l$, and the second from the fact that $\tilde{R}^{R \rightarrow RA}$ is isometry.
We then have
\begin{widetext}
    \begin{align}
        \bigl\| \cQ_F^{CA \rightarrow A} \circ \tilde{\cR}_E^{R\rightarrow RA}(\Phi_{\cT}^{CR}) - \Psi^{AR} \bigr\|_1
        &\leq \sum_l q(l) \bigl\| q(l)^{-1}\cQ_{F,l}^{CA \rightarrow A} \circ \tilde{\cR}_E^{R\rightarrow RA}(\Phi_{\cT}^{CR}) - \Psi_l^{AR} \bigr\|_1\\
        &\leq 2 \sum_l q(l) \sqrt{1 - q(l)^{-1} \bra{\Psi_l}^{AR} \cQ_{F,l}^{CA \rightarrow A} \circ \tilde{\cR}_E^{R\rightarrow RA}(\Phi_{\cT}^{CR}) \ket{\Psi_l}^{AR}}\\
        &\leq 2 \sqrt{\sum_l  q(l) \Bigl( 1 - q(l)^{-1} \bra{\Psi_l}^{AR} \cQ_{F,l}^{CA \rightarrow A} \circ \tilde{\cR}_E^{R\rightarrow RA}(\Phi_{\cT}^{CR}) \ket{\Psi_l}^{AR}\Bigr)}\\
        &= 2 \sqrt{1 - \sum_l \bra{\Psi_l}^{AR} \cQ_{F,l}^{CA \rightarrow A} \circ \tilde{\cR}_E^{R\rightarrow RA}(\Phi_{\cT}^{CR}) \ket{\Psi_l}^{AR}},
        \label{Eq:lcji0verov}
    \end{align}
\end{widetext}
where the convexity of the trace norm is used for the first inequality, the Fuchs-van de Graaf inequality (\cref{FvG}) for the second, and the concavity of a square-root function for the last. Note that $\{ q(l) \}_l$ is a probability distribution.

We explicitly compute the inner product in the summation.
Using an adjoint map $\cE^{\dagger}$ of a quantum channel $\cE$, defined by $\tr[\cE(X)Y] = \tr[X\cE^{\dagger}(Y)]$, we have
\begin{multline}
\bra{\Psi_l}^{AR} \cQ_{F,l}^{CA \rightarrow A} \circ \tilde{\cR}_E^{R\rightarrow RA}(\Phi_{\cT}^{CR}) \ket{\Psi_l}^{AR}\\
=
\tr\bigl[\Phi_{\cT}^{CR}  (\cQ_{F,l}^{CA \rightarrow A} \circ \tilde{\cR}_E^{R\rightarrow RA}\bigr)^{\dagger}(\Psi_l^{AR}) \bigr]. \label{Eq:coergn}
\end{multline}
More concretely, as $\tilde{R}_E^{R\rightarrow RA}$ is an isometry, its adjoint is given by 
\begin{equation}
    \tilde{\cR}_E^{R \rightarrow RA \dagger}(\rho^{AR}) = \tilde{R}^{R \rightarrow RA \dagger}_E \rho^{AR} \tilde{R}^{R \rightarrow RA}_E.
\end{equation}
From~\cref{Eq:defC}, it is straightforward to see 
\begin{equation}
\cQ_{F,l}^{CA \rightarrow A \dagger}(\rho^A) = \bigl(\Theta_l^{A \dagger} \rho^{A} \Theta_l^{A} \bigr) \otimes M_{F,l}^C.
\end{equation}
From these explicit forms, we observe
\begin{multline}
    (\cQ_{F,l}^{CA \rightarrow A} \circ \tilde{\cR}_E^{R\rightarrow RA}\bigr)^{\dagger}(\Psi_l^{AR}) \\
    =
    \tilde{R}^{R \rightarrow RA \dagger}_E \Theta_l^{A \dagger} \Psi_l^{AR} \Theta_l^{A} \tilde{R}^{R \rightarrow RA}_E \otimes M_{F,l}^C.
\end{multline}
Recalling that $\Psi_l^{AR}$ is pure, we obtain
\begin{align}
    \tilde{R}^{R \rightarrow RA \dagger}_E \Theta_l^{A \dagger} \ket{\Psi_l}^{AR}
    &= \sum_j \bigl| \braket{l_F}{j_E}  \bigr| \bra{j_E} \Theta_l^{\dagger} \ket{j_E} \ket{j^*_E}^R\\
    &= \sum_j \bigl| \braket{l_F}{j_E} \bigr| e^{-i {\rm arg} \braket{j_E}{l_F}}  \ket{j^*_E}^R\\
    &= \sum_j \braket{j_E^*}{l_F^*}  \ket{j^*_E}^R\\
    &= \ket{l_F^*}^R,
\end{align}
where we used the definition of $\Theta_l$, given by~\cref{Eq:CorrTheta}, to obtain the second equality.
We hence have 
\begin{equation}
    (\cQ_{F,l}^{CA \rightarrow A} \circ \tilde{\cR}_E^{R\rightarrow RA}\bigr)^{\dagger}(\Psi_l^{AR})
    = \ketbra{l_F^*}{l_F^*}^R \otimes M_{F, l}^C,
\end{equation}
which, together with~\cref{Eq:coergn}, leads to
\begin{align}
    &\bra{\Psi_l}^{AR} \cQ_{F,l}^{CA \rightarrow A} \circ \tilde{\cR}_E^{R\rightarrow RA}(\Phi_{\cT}^{CR}) \ket{\Psi_l}^{AR} \notag\\
    &=\tr \bigl[ \Phi_{\cT}^{CR} (\ketbra{l_F^*}{l_F^*}^R \otimes M_{F, l}^C) \bigr].
\end{align}

As $\bra{l_F^*}^R\Phi_{\cT}^{AR} \ket{l_F^*}^R = \cT^{A \rightarrow C}(\ketbra{l_F}{l_F}^A)/d$, we arrive at
\begin{multline}
    \bra{\Psi_l}^{AR} \cQ_{F,l}^{CA \rightarrow A} \circ \tilde{\cR}_E^{R\rightarrow RA}(\Phi_{\cT}^{CR}) \ket{\Psi_l}^{AR}\\
    =\frac{1}{d}\tr \bigl[ \cT^{A \rightarrow C}(\ketbra{l_F}{l_F}^A) M_{F, l}^C \bigr].
\end{multline}
Substituting this into~\cref{Eq:lcji0verov}, we have
\begin{align}
    &\bigl\| \cQ_F^{CA \rightarrow A} \circ \tilde{\cR}_E^{R\rightarrow RA}(\Phi_{\cT}^{CR}) - \Psi^{AR} \bigr\|_1 \notag \\
    &\leq 
    2 \sqrt{1 - \frac{1}{d}\sum_l \tr\bigl[ \cT^{A \rightarrow C}\bigl( \ketbra{l_F}{l_F}^A \bigr) M_{F, l}^C \bigr]}\\
    &\leq 2 \sqrt{\Delta_{{\rm cl}, F}}. \label{Eq:asjif;jew42t}
\end{align}

Finally, we show that $\Psi^{AR} \approx \Phi^{AR}$ if $(E, F)$ is close to MUBs.
We begin with
\begin{align}
    \| \Psi^{AR} - \Phi^{AR} \|_1 &\leq 2 \sqrt{1 - \bra{\Phi}\Psi\ket{\Phi}}.
\end{align}
It is easily shown that
\begin{align}
    \bra{\Phi}\Psi\ket{\Phi}
    &=\sum_l q(l) |\braket{\Psi_l}{\Phi}|^2\\
    &=\sum_l q(l) \bigl( d^{-1/2} \sum_{j} |\braket{j_E}{l_F}| \bigr)^2\\
    &=\sum_l q(l) F_{BC}({\rm unif}_d, p_l),
\end{align}
where $p_l(j) = |\braket{j_E}{l_F}|^2$.
Recalling that $q(l) = \tr\bigl[\cT^{A \rightarrow C}(\pi^A) M_{F, l}^C \bigr]$, we obtain
\begin{align}
    \| \Psi^{AR} - \Phi^{AR} \|_1 &\leq 2 \sqrt{\Xi_{EF}}. \label{Eq:korefeko}
\end{align}

The statement of~\cref{Thm:CtoQ} follows simply from the triangle inequality:
\begin{widetext}
    \begin{align}
        \Delta_{{\rm cl}}(\cD_{\rm CtoQ}|\cT) 
        &=\frac{1}{2}\| \cQ_F^{CA\rightarrow A}\circ \cR_{E}^{C\rightarrow CA}(\Phi_{\cT}^{CR}) - \Phi^{AR} \bigr\|_1 \\
        \begin{split}
        &\leq \frac{1}{2} \bigl\|  \cQ_F^{CA \rightarrow A} \circ\bigl( \cR_E^{C \rightarrow CA} - \tilde{\cR}_{E}^{R\rightarrow RA}\bigr) (\Phi_{\cT}^{CR}) \bigr\|_1 +  \frac{1}{2}\| \cQ_F^{CA\rightarrow A}\circ  \tilde{\cR}_E^{R\rightarrow RA}(\Phi_{\cT}^{CR}) - \Psi^{AR} \|_1 \\
        & \hspace{2cm} +
        \frac{1}{2} \|\Psi^{AR} - \Phi^{AR}\|_1 
        \end{split} \\
        \begin{split}
        &\leq \frac{1}{2} \bigl\|  \bigl( \cR_E^{C \rightarrow CA} - \tilde{\cR}_{E}^{R\rightarrow RA}\bigr) (\Phi_{\cT}^{CR}) \bigr\|_1 +  \frac{1}{2}\| \cQ_F^{CA\rightarrow A}\circ  \tilde{\cR}_E^{R\rightarrow RA}(\Phi_{\cT}^{CR}) - \Psi^{AR} \|_1 \\
        & \hspace{2cm} +
        \frac{1}{2} \|\Psi^{AR} - \Phi^{AR}\|_1 
        \end{split} \\
        &\leq
        \sqrt{\Delta_{{\rm cl},E}(2-\Delta_{{\rm cl},E})} + \sqrt{\Delta_{{\rm cl},F}} + \sqrt{\Xi_{EF}},
    \end{align}
\end{widetext}
where the second inequality follows from the monotonicity of the trace distance and the last one from~\cref{Eq:Xerror,Eq:asjif;jew42t,Eq:korefeko}. 
\end{proof}

\subsection{Proof of~\cref{Prop:SC}} \label{S:pPGMs}

\SC*

In the proof, we do not explicitly write the superscript, such as $A$ and $C$, nor the basis $W$ in $\Delta_{{\rm cl}, W}$.
The pPGM is given by
\begin{equation}
M_{\rm pPGM} = \bigl\{M_j:= \Pi^{-1/2} \Pi_j \Pi^{-1/2} \bigr\}_j,
\end{equation}
where $\Pi_j$ is a projection onto the support of $\tau_j^C$ and $\Pi := \sum_j \Pi_j$.

\begin{proof}[Proof of~\cref{Prop:SC}]
    The error in decoding classical information is rewritten as
    \begin{align}
        &\Delta_{{\rm cl}} (M_{\rm pPGM}|\cN \circ \cE) = \frac{1}{d} \sum_{j} \tr\bigl[ ( I - M_j)\tau_j  \bigr]. 
    \end{align}
    In~\cref{App:Holevo}, we follow the original analysis by Holevo~\cite{H1998} and by Schumacher and Westmoreland~\cite{SW1998} and show that 
    \begin{equation}
        \Delta_{{\rm cl}}(M_{\rm pPGM}|\cN \circ \cE) \leq \frac{1}{d} \sum_{i \neq j} \tr \bigl[ \Pi_i \tau_j  \bigr]. \label{Eq:HN}
    \end{equation}
    Since $\lambda_{\rm min} (\tau_i) \Pi_i \leq \tau_i$, where $\lambda_{\rm min} (\tau_i)$ is the minimum non-zero eigenvalue of $\tau_i$, we have 
    \begin{align}
        \Delta_{{\rm cl}}(M_{\rm pPGM}|\cN \circ \cE)
        &\leq \frac{1}{d} \sum_{i \neq j} \frac{\tr \bigl[ \tau_i \tau_j \bigr]}{ \lambda_{\rm min} (\tau_i)} \\
        &\leq \frac{1}{d \lambda_{\rm min}} \sum_{i \neq j} \tr \bigl[ \tau_i \tau_j \bigr] \\
        &= \frac{1}{\lambda_{\rm min}} \biggl\{ d \tr \bigl[  \tau_{\pi}^2 \bigr] -  \frac{1}{d}\sum_{j}\tr \bigl[ \tau_j^2 \bigr] \biggr\},
    \end{align}
    where $\lambda_{\rm min} = \min_j \lambda_{\rm min} (\tau_j)$ and we used $\tau_{\pi} = d^{-1}\sum_j \tau_j$. This completes the proof.
\end{proof}

\section{Decoding errors in the Hayden-Preskill protocol} \label{S:ErrorHP}

We provide an in-depth analysis of decoding the Hayden-Preskill protocol by the C-to-Q decoder with pPGMs. 
To this end, we introduce a family of the Hayden-Preskill protocol in~\cref{SS:FHP}, which is to clearly understand different asymptotic limits of the protocol, and provide the formal statement of~\cref{Thm:HPformal}.
We then provide a single-shot analysis in~\cref{SS:rrccr}. Based on these, we provide a proof of~\cref{Thm:HPformal} in~\cref{SS:owari}.

\subsection{A family of the Hayden-Preskill protocol} \label{SS:FHP}

To fully understand the protocol, we introduce a sequence of Hayden-Preskill protocols labeled by $n=1,2,\ldots$. This allows us to argue different asymptotic limits simultaneously. Note that the Hayden-Preskill protocol has multiple independent parameters, such as $k$, $N$, $\ell$, and $\xi$, and hence, various asymptotic limits can be considered.
We denote the parameters for the $n$-th protocol by $(N_n, k_n, \ell_n, \xi_n)$ and introduce
\begin{align}
    &\ell_{{\rm th}, n} = k_n + \frac{N_n - H_2(B_{\rm in})_{\xi_n}}{2},\\
    & \Lambda_{\xi_n} = {\rm rank}(\xi^{B_{\rm in}}_n) \lambda_{\rm min}(\xi^{B_{\rm in}}_n) \in (0, 1].
\end{align}

From the previous result based on the decoupling approach (\cref{Thm:HP}), it immediately follows that, for every $n$, there exist a decoding map $\cD_n$ and decoding POVMs $M_{W,n}$ for any basis $W$ such that
\begin{multline}
    \mbb{E}_{U_n \sim {\sf H}_n} \bigl[ \Delta_{{\rm cl}, W}(M_{W,n}|\xi_n, U_n) \bigr]
    \le  \\
    \mbb{E}_{U_n \sim {\sf H}_n} \bigl[ \Delta_{\rm{q}}(\cD_n|\xi_n, U_n) \bigr] \le 2^{(\ell_{\mathrm{th},n} - \ell_n )/2}.
\end{multline}

In the $n$-th protocol, we consider decoding $W$-classical information by pPGMs $M_{{\rm pPGM}, W, n}$ ($W = Z, X$) and quantum information by the C-to-Q decoder $\cD_{{\rm CtoQ},n}$ constructed from the pPGM for decoding the $Z$- and $X$-classical information. These decoders are all dependent on $\xi_n$ as well as $U_n$.
The decoding errors are denoted by $\Delta_{{\rm cl}, W}(M_{{\rm pPGM},W,n}|\xi_n, U_n)$ for $W$-classical information and by $\Delta_{\rm{q}}(\cD_{{\rm CtoQ},n}|\xi_n, U_n)$ for quantum information. For simplicity, we hereafter omit the basis $W$ in the subscript of the pPGM $M_{{\rm pPGM}, W, n}$ when we refer to the classical decoding error, such as $\Delta_{{\rm cl}, W}(M_{{\rm pPGM},n}|\xi_n, U_n)$.

Taking the average over a Haar random unitary $U_n \sim {\sf H}_n$, we define decoding errors on average:
\begin{align}
    &\overline{\Delta_{\mathrm{cl},n }} = \mbb{E}_{U_n \sim {\sf H}_n} \bigl[ \Delta_{{\rm cl}, W}(M_{{\rm pPGM},n}|\xi_n, U_n) \bigr],\\
    &\overline{\Delta_{\mathrm{q},n }} = \mbb{E}_{U_n \sim {\sf H}_n} \bigl[ \Delta_{\rm{q}}(\cD_{{\rm CtoQ},n}|\xi_n, U_n) \bigr].
\end{align}
Note that, due to the unitary invariance of the Haar measure, the average decoding error $\overline{\Delta_{\mathrm{cl},n }}$ for $W$-classical information does not depend on the choice of the basis.

The main result about the Hayden-Preskill protocol is given as follows.\\

\noindent
\textbf{\cref{Thm:HPformal}.}
\emph{For any sequence of Hayden-Preskill protocols that satisfies $\epsilon := \limsup_{n\to\infty} \{ 2(1- \Lambda_{\xi_n}  )  \} <1 $ and
    \begin{equation}
        \lim_{n\to\infty} (N_n+k_n-\ell_n)2^{- (k_n + 2(\ell_n - \ell_{\mathrm{th},n} ))} = 0, \label{Eq:ccccceeer}
    \end{equation}
    it holds that
    \begin{equation}
        \limsup_{n\to\infty} \left\{ \overline{\Delta_{\mathrm{cl},n}} \biggl(\frac{2^{2(\ell_{\mathrm{th},n} - \ell_n )}}{1-\epsilon}\biggr)^{-1}
        \right\} \le 1, \label{Eq:ccccc}
    \end{equation}
    and
    \begin{equation}
        \limsup_{n\to\infty} \left\{ \overline{\Delta_{\mathrm{q},n}} \biggr(\sqrt{\frac{2^{2(\ell_{\mathrm{th},n} - \ell_n )}}{1-\epsilon}}\biggr)^{-1}
        \right\} \le 1+\sqrt{2}. \label{Eq:qqqqqqqq}
    \end{equation}
}\\

This can be obtained by the two steps. We provide a single-shot analysis in~\cref{SS:rrccr} and prove~\cref{Thm:HPformal} in~\cref{SS:owari}.

\subsection{A single-shot analysis on the decoding errors} \label{SS:rrccr}

In this subsection, we provide a proposition about the decoding errors in the Hayden-Preskill protocol for a fixed $n$. We hence omit the labeling $n$ of the family.

\begin{proposition} \label{Thm:result3}
Suppose that $\Lambda_{\xi}:={\rm rank}(\xi^{B_{\rm in}}) \lambda_{\rm min}(\xi^{B_{\rm in}})>1/2$, where $\lambda_{\rm min}$ denotes the minimum non-zero eigenvalue.
For any constants $\epsilon \in \bigl(2 (1 - \Lambda_{\xi} ), 1 \bigr]$, it holds for any basis $W$ that
\begin{align}
&\EH \bigl[ \Delta_{{\rm cl}, W}(M_{{\rm pPGM},W}|\xi, U) \bigr]  \leq
\frac{4^{\ell_{\rm th}-\ell} }{1-\epsilon} + \delta, \\
&\EH[\Delta_{\rm{q}}(\cD_{\rm CtoQ}|\xi, U)]  \leq  \bigl(1 + \sqrt{2}\bigr)\sqrt{\frac{4^{\ell_{\rm th}-\ell} }{1-\epsilon} + \delta},
\end{align}
where $\ell_{\rm th} = k+ \frac{N-H_2(B_{\rm in})_{\xi}}{2}$, $\delta$ is given by
\begin{multline}
    \log \delta
    :=
    k + 2^{N+k-\ell+1} \biggl(N + k - \ell+ \log \frac{5}{\epsilon} \biggr) \\
- \frac{c^2 \log e}{6}2^{\ell + H_2(B_{\rm in})_{\xi}},
\end{multline}
and $c=1 - (1-\epsilon/2)/\Lambda_{\xi}$.
\end{proposition}

To show~\cref{Thm:result3}, we first show a general property about the minimum eigenvalue of a marginal state after the application of a Haar random unitary. In~\cref{Prop:conc}, the labeling of the systems, $A$ and $B$, are general and not those in the Hayden-Preskill protocol.

\begin{proposition} \label{Prop:conc}
Let $AB$ be a composite system, and $d_A$ and $d_B$ be the dimensions of $A$ and $B$, respectively. For a mixed state $\rho^{AB}$ with rank $r$ and a Haar random unitary  $U^{AB}$, let $\rho_U^{B}$ be $\tr_A[\cU^{AB}(\rho^{AB})]$.
For any $\epsilon>2 (1 - r \lambda_{\rm min}(\rho^{AB}))$, 
\begin{multline}
{\rm Prob}_{U \sim \sf H}\biggl[ \lambda_{\rm min} \bigl( \rho_U^{B} \bigr) < \frac{1- \epsilon}{d_B} \biggr] \\
\leq \biggl( \frac{5 d_B}{\epsilon} \biggr)^{2 d_B} \exp\biggl[ - \frac{r d_A \delta_{\rho}(\epsilon/2)^2}{6} \biggr].
\end{multline}
Here, $\delta_{\rho}(\epsilon):= 1 - \frac{1- \epsilon}{r \lambda_{\rm min}(\rho^{AB})}$.
\end{proposition}

\cref{Prop:conc} is a generalization of Lemma III.4 in Ref.~\cite{HLW2006} to the case of mixed states and can be shown by the same proof technique using the \emph{$\epsilon$-net}~\cite{HLSW2004}.

\begin{definition}[$\epsilon$-net]
For $\epsilon>0$, a set $N(\epsilon)$ of pure states in a Hilbert space $\cH$ is called an $\epsilon$-net if
\begin{multline}
\forall \ket{\varphi} \in \cH,\ \ \exists \ket{\varphi'} \in N(\epsilon),\\ \text{such that}\ \ \bigl\| \ketbra{\varphi}{\varphi} - \ketbra{\varphi'}{\varphi'} \bigr\|_1 \leq \epsilon.
\end{multline}
\end{definition}

It is known that there exists an $\epsilon$-net with a sufficiently many but finite number of pure states.

\begin{theorem}\label{Thm:net}
There exists an $\epsilon$-net $N(\epsilon)$ such that $|N(\epsilon)| \leq ( 5/\epsilon )^{2 d}$, where $d := \dim \cH$.
\end{theorem}

We also use a concentration of measure phenomena for a Haar random unitary.

\begin{lemma}[Lemma III.5 in~\cite{HLW2006}]\label{Lem:CoM}
For projections $S, Q$ on a $d$-dimensional Hilbert space $\cH$ and $\epsilon>0$, it follows that 
\begin{equation}
{\rm Prob}_{ U \sim {\sf H}}\biggl[ \tr\bigl[USU^{\dagger}Q \bigr] < (1-\epsilon)\frac{sq}{d} \biggr]
\leq \exp \biggl[ -\frac{s q \epsilon^2}{6} \biggr],
\end{equation}
where $s$ and $q$ are the rank of $S$ and $Q$, respectively.
\end{lemma}

Based on~\cref{Thm:net} and~\cref{Lem:CoM}, we show~\cref{Prop:conc}. The proof is almost the same as that in~\cite{HHL2004}.

\begin{proof}[Proof of~\cref{Prop:conc}]
By definition, the minimum non-zero eigenvalue is given by
$\lambda_{\rm min}(\rho_U^B)
=
\min_{\ket{\varphi} \in_+ \rho_U} \bra{\varphi} \rho_U^B \ket{\varphi}$,
where $\min_{\ket{\varphi} \in_+ \rho_U}$ represents the minimization over all pure states $\ket{\varphi}^B$ in the support of $\rho_U^B$.
From~\cref{Thm:net}, there exists an $\epsilon/d_B$-net $N_B(\epsilon/d_B)$ of the Hilbert space $\cH^B$ with $|N_B(\epsilon/d_B)| \leq (5 d_B/\epsilon)^{2d_B}$. Hence, we have
\begin{equation}
\lambda_{\rm min}(\rho_U^B)
\geq
\min_{\ket{\varphi} \in N_B(\epsilon/d_B)} \bra{\varphi} \rho_U^B \ket{\varphi} - \frac{\epsilon}{2 d_B}.
\end{equation}
Here, we used the fact that, if $\| \ketbra{\psi}{\psi} -\ketbra{\psi'}{\psi'}\|_1 \leq \epsilon$, then $| \bra{\psi} \sigma \ket{\psi} - \bra{\psi'} \sigma \ket{\psi'} | \leq \epsilon/2$ for any state $\sigma$.
This implies 
\begin{multline}
{\rm Prob}_{ U \sim {\sf H}}\biggl[ \lambda_{\rm min} \bigl( \rho_U^{B} \bigr) < \frac{1- \epsilon}{d_B} \biggr]\\ \leq 
{\rm Prob}_{ U \sim {\sf H}}\biggl[
\min_{\ket{\varphi} \in N_B(\epsilon/d_B)} \bra{\varphi} \rho_U^B \ket{\varphi} < \frac{1 - \epsilon/2}{d_B}
\biggr]. \label{Eq:1111111}
\end{multline}

Decomposing $\rho^{AB}$ into $\sum_{j=1}^r \lambda_j \ketbra{\psi_j}{\psi_j}^{AB}$, where $\ket{\psi_j}$ are eigenstates and $\lambda_j$ are non-zero eigenvalues, and using the projection $R^{AB} = \sum_{j=1}^r \ketbra{\psi_j}{\psi_j}^{AB}$, it holds for any $\ket{\varphi} \in \cH^B$ that
\begin{align}
    &\bra{\varphi} \rho_U^B \ket{\varphi} \notag \\
    &=
    \sum_{j=1}^r \lambda_j \tr \bigl[ (I^A \otimes \ketbra{\varphi}{\varphi}^B) U^{AB} \ketbra{\psi_j}{\psi_j}^{AB} (U^{AB})^{\dagger} \bigr]\\
    &\geq
    \lambda_{\rm min}(\rho^{AB}) \tr \bigl[ (I^A \otimes \ketbra{\varphi}{\varphi}^B) U^{AB} R^{AB} (U^{AB})^{\dagger} \bigr].
\end{align}
Hence, for any $\ket{\varphi}^B$ and $\epsilon>2 (1 - r \lambda_{\rm min}\bigl(\rho^{AB})\bigr)$, 
\begin{align}
    &{\rm Prob}_{U \sim {\sf H}}\biggl[ \bra{\varphi} \rho_U^B \ket{\varphi}  < \frac{1 - \epsilon/2}{d_B} \biggr] \notag \\
    &\leq 
    {\rm Prob}_{U \sim {\sf H}}\biggl[
    \tr \bigl[ (I^A \otimes \ketbra{\varphi}{\varphi}^B) U^{AB} R^{AB} (U^{AB})^{\dagger} \bigr] \notag \\
    &\hspace{30mm}< \bigl( 1- \delta_{\rho}(\epsilon/2) \bigr)\frac{r}{d_B} \biggr]\\
    &\leq 
    \exp \biggl[ - \frac{d_A r \delta_{\rho}(\epsilon/2)^2}{6} \biggr], \label{Eq:2222222}
\end{align}
where $\delta_{\rho}(x) := 1 - (1- x)/(r \lambda_{\rm min}(\rho^{AB}))$ and the last line follows from~\cref{Lem:CoM}.

From~\cref{Eq:1111111,Eq:2222222}, and using the union bound, we obtain
\begin{multline}
    {\rm Prob}_{\sf H}\biggl[ \lambda_{\rm min} \bigl( \rho_U^{B} \bigr) < \frac{1- \epsilon}{d_B} \biggr] \\
    \leq 
    \biggl( \frac{5 d_B}{\epsilon} \biggr)^{2 d_B}\exp \biggl[ - \frac{d_A r \delta_{\rho}(\epsilon/2)^2}{6} \biggr],
\end{multline}
for any $\epsilon>2 (1 - r \lambda_{\rm min}\bigl(\rho^{AB})\bigr)$, which concludes the proof.
\end{proof}

\begin{proof}[Proof of~\cref{Thm:result3}]

For simplicity, we introduce 
\begin{equation}
    \xi_{j}^{B_{\rm rad} S}(U) = \cU^S \bigl( \ketbra{j_W}{j_W}^A \otimes \xi^{B_{\rm in}B_{\rm rad}} \bigr).
\end{equation}
This is the state after the application of a unitary $U^S$ when the initial state in $A$ is $\ket{j_W}^A$. Since $\xi^{B_{\rm in} B_{\rm rad}}$ is a pure state, so is $\xi_j^{B_{\rm rad} S}(U)$.

Using~\cref{Prop:SC}, the decoding error for the $W$-classical information, when it is decoded by the corresponding pPGM, is given by
\begin{multline}
\Delta_{{\rm cl}, W}(M_{{\rm pPGM},W}|\xi, U) \\
\leq 
\frac{1}{2^k \lambda_{\rm min}} \sum_{i \neq j} \tr \bigl[ \xi_{i}^{B_{\rm rad} S_{\rm rad}}(U)\xi_{j}^{B_{\rm rad} S_{\rm rad}}(U) \bigr], 
\label{Eq:pPGMError}
\end{multline}
where 
\begin{align}
\lambda_{\rm min} &= \min_j \{ \lambda_{\rm min} \bigl(\xi_{j}^{B_{\rm rad} S_{\rm rad}}(U) \bigr) \}\\ 
&= \min_j \{ \lambda_{\rm min} \bigl(\xi_{j}^{S_{\rm in}}(U) \bigr) \}, \label{Eq:minimini}
\end{align}
with $\lambda_{\rm min}(\sigma)$ being the minimum non-zero eigenvalue of $\sigma$. The second equality holds since $\xi_j^{B_{\rm rad}S}$ is a pure state.

To investigate $\lambda_{\rm min} (\xi_{j}^{S_{\rm in}}(U) )$, we apply~\cref{Prop:conc} with the following identification: $A$ and $B$ in~\cref{Prop:conc} corresponds to $S_{\rm rad}$ and $S_{\rm in}$, respectively. Then, for any $\epsilon$ such that $\epsilon \geq 2(1 - r \lambda_{\rm min}(\xi^{B_{\rm in}}))$, where $r = {\rm rank} (\ketbra{j}{j}^A \otimes  \xi^{B_{\rm in}}) = {\rm rank} (\xi^{B_{\rm in}})$, we have $\lambda_{\rm min}(\xi_j^{S_{\rm in}}(U)) < \frac{1-\epsilon}{2^{N+k-\ell}}$ with probability at most
\begin{equation}
P_{\xi}(\epsilon) := \biggl( \frac{5 \cdot 2^{N+k-\ell}}{\epsilon} \biggr)^{2^{N+k-\ell+1}} \exp \biggl[ -\frac{2^{\ell} r c^2}{6} \biggr]. \label{Eq:44vb9sp}
\end{equation}

For $\epsilon$ such that $\epsilon \geq 2(1 - r \lambda_{\rm min}(\xi^{B_{\rm in}}))$, we define $\mathfrak{U}_{\epsilon}$ as
\begin{equation}
\mathfrak{U}_{\epsilon} := \biggl\{ U \in \mathfrak{U}(2^{N+k}) : \lambda_{\rm min} \geq \frac{1- \epsilon}{2^{N+k-\ell}} \biggr\}. \label{Eq:uep}
\end{equation}
Note that $\lambda_{\rm min}$ depends on $U$ (see~\cref{Eq:minimini}).
Using~\cref{Eq:44vb9sp} and the union bound over the choice of $j \in \{0, \dots, 2^k-1\}$, the set $\mathfrak{U}_{\epsilon}$ satisfies
\begin{equation}
{\sf H}(\mathfrak{U}_{\epsilon}) \geq 1 - 2^k P_{\xi}(\epsilon). \label{Eq:ProbLmin}
\end{equation}

We denote by ${\sf H}_{\epsilon}$ and $\bar{{\sf H}}_{\epsilon}$ the probability measures induced from the Haar measure ${\sf H}$ by restriction to $\mathfrak{U}_{\epsilon}$ and to the complementary set of $\mathfrak{U}_{\epsilon}$ in $\mathfrak{U}$, respectively.
Using $\EH = {\sf H}(\mathfrak{U}_{\epsilon}) \mbb{E}_{U \sim {\sf H}_{\epsilon}} + \bigl( 1 - {\sf H}(\mathfrak{U}_{\epsilon}) \bigr) \mbb{E}_{U \sim \bar{\sf H}_{\epsilon}}$, we have
\begin{multline}
\EH \bigl[ \Delta_{{\rm cl}, W}(M_{{\rm pPGM},W}|\xi, U) \bigr] \\
\leq 
{\sf H}(\mathfrak{U}_{\epsilon})  \mbb{E}_{U \sim {\sf H}_{\epsilon}} \bigl[\Delta_{{\rm cl}, W}(M_{{\rm pPGM},W}|\xi, U) \bigr]+ 2^k P_{\xi}(\epsilon),
\end{multline}
where we have used~\cref{Eq:ProbLmin} and that $0\leq \Delta_{{\rm cl}}(M_{\rm pPGM}|\xi, U) \leq 1$.
Further using~\cref{Eq:pPGMError,Eq:uep}, we have
\begin{align}
&{\sf H}(\mathfrak{U}_{\epsilon})  \mbb{E}_{U \sim {\sf H}_{\epsilon}} \bigl[\Delta_{{\rm cl}, W}(M_{{\rm pPGM},W}|\xi, U) \bigr] \notag \\
&\leq 
\frac{2^{N - \ell}}{1- \epsilon} 
{\sf H}(\mathfrak{U}_{\epsilon})  \mbb{E}_{U \sim {\sf H}_{\epsilon}}
 \biggl[ 
\sum_{i \neq j} \tr \bigl[ \xi_{i}^{B_{\rm rad} S_{\rm rad}}(U)\xi_{j}^{B_{\rm rad} S_{\rm rad}}(U) \bigr]
\biggr] \label{Eq:rthispdojv/pr}\\
&\leq 
\frac{2^{N - \ell}}{1- \epsilon} E,
\end{align}
where we defined
\begin{equation}
E:=\mbb{E}_{U \sim {\sf H}}  \biggl[ 
\sum_{i \neq j} \tr \bigl[ \xi_{i}^{B_{\rm rad} S_{\rm rad}}(U)\xi_{j}^{B_{\rm rad} S_{\rm rad}}(U) \bigr]
\biggr]. 
\end{equation}
Note that~\cref{Eq:rthispdojv/pr} follows from the fact that $\lambda_{\rm min} \geq (1-\epsilon)/2^{N+k-\ell}$ in $\mathfrak{U}_{\epsilon}$, on which the probability measure ${\sf H}_{\epsilon}$ is defined.
We, hence, arrive at
\begin{equation}
\EH \bigl[ \Delta_{{\rm cl}, W} (M_{{\rm pPGM},W}|\xi, U) \bigr]
\leq 
\frac{2^{N - \ell}}{1- \epsilon} E + 2^k  P_{\xi}(\epsilon), \label{Eq:c34tpgrme}
\end{equation}
for $\epsilon>2 (1 - r \lambda_{\rm min}\bigl(\xi^{B_{\rm in}})\bigr)$.

In the following, we focus on $E$, which can be computed using the swap trick.
We use the notation that, for a given system $X$, $X'$ denotes a copied system isomorphic to $X$. We also use the swap operator $\mbb{F}^{X, X'} = \sum_{i,j} \ketbra{e_i}{e_j}^X \otimes \ketbra{e_j}{e_i}^{X'}$, where $\{ \ket{e_i} \}_i$ is an orthonormal basis in $X$. A well-known property of the swap operator is that, for two operators $\omega$ and $\gamma$ on $X$, we have $\tr[\omega \gamma] = \tr[ (\omega^X \otimes \gamma^{X'})\mbb{F}^{X, X'}]$.

Using the swap trick in $B_{\rm rad} S_{\rm rad}$, which together we denote by ${\rm rad}$, we have the following identity relation:
\begin{align}
&\tr \bigl[ \xi_{i}^{\rm rad}(U) \xi_{j}^{\rm rad}(U) \bigr] \notag \\
&=
\tr \bigl[ \bigl( \xi_{i}^{\rm rad}(U) \otimes \xi_{j}^{{\rm rad}'}(U) \bigr)  \mbb{F}^{{\rm rad}, {\rm rad}'} \bigr]\\
&=
\tr \bigl[ \bigl( \xi_{i}^{SB_{\rm rad}}(U) \otimes \xi_{j}^{S'B_{\rm rad}'}(U) \bigr)  \bigl( \mbb{I}^{S_{\rm in} S'_{\rm in}} \otimes \mbb{F}^{{\rm rad}, {\rm rad}'}\bigr)  \bigr],
\end{align}
where $\xi_{i}^{SB_{\rm rad}}(U)$ and $\xi_{j}^{S'B_{\rm rad}'}(U)$ are the pure states defined as $\ket{\xi_i}^{SB_{\rm rad}} = U^S \bigl( \ket{i}^A \otimes \ket{\xi}^{B_{\rm in}B_{\rm rad}} \bigr)$ and so on.
Using this, we have
\begin{equation}
    E = \sum_{i \neq j} \tr \bigl[ O_{ij} \bigl( \mbb{I}^{S_{\rm in} S'_{\rm in}} \otimes \mbb{F}^{{\rm rad}, {\rm rad}'}\bigr)  \bigr], \label{Eq:E}
\end{equation}
where we defined an average operator $O_{ij}$ on $SS'B_{\rm rad}B_{\rm rad}'$ as
\begin{equation}
O_{ij} := \mbb{E}_{U \sim {\sf H}} \bigl[ \xi_{i}^{SB_{\rm rad}}(U) \otimes \xi_{j}^{S'B_{\rm rad}'}(U) \bigr]
\end{equation}
for $i\neq j$.

To compute $O_{ij}$, let
\begin{equation}
\ket{\xi}^{B_{\rm in} B_{\rm rad}}
=
\sum_{m} \sqrt{\xi_m} \ket{\varphi_m}^{B_{\rm in}} \otimes \ket{\psi_m}^{B_{\rm rad}},
\end{equation}
be the Schmidt decomposition of the pure state $\ket{\xi}$.
We introduce
\begin{multline}
O_{imn}^{jm'n'}:=\mbb{E}_{U \sim {\sf H}} \bigl[(U^S)^{\otimes 2}  \bigl( \ketbra{j}{j}^A \otimes \ketbra{i}{i}^{A'} \\
\otimes \ketbra{\varphi_m}{\varphi_n}^{B_{\rm in}} \otimes \ketbra{\varphi_{m'}}{\varphi_{n'}}^{B_{\rm in}'} \bigr) (U^{S \dagger})^{\otimes 2} \bigr],
\end{multline}
for $i \neq j$.
This is related to $O_{ij}$ as 
\begin{multline}
    O_{ij} = \sum_{m,n,m',n'} \sqrt{\xi_m \xi_{m'}\xi_n\xi_{n'}} \\
    O_{imn}^{jm'n'} \otimes \ketbra{\psi_m}{\psi_n}^{B_{\rm rad}}\otimes \ketbra{\psi_{m'}}{\psi_{n'}}^{B_{\rm rad}'}. \label{Eq:crrrrrrleee}
\end{multline}

The operator $O_{imn}^{jm'n'}$ commutes with any unitary in the form of $V^S \otimes V^{S'}$ due to the unitary invariance of the Haar measure ${\sf H}$. From the Schur's lemma, this implies that the operator is a linear combination of the projection onto the symmetric subspace and that onto the anti-symmetric subspace. The projections to the former and the latter subspaces are given by $\frac{\mbb{I}^{SS'} + \mbb{F}^{SS'}}{2}$ and $\frac{\mbb{I}^{SS'} - \mbb{F}^{SS'}}{2}$, respectively. Further using $\tr[O_{imn}^{jm'n'}] = \delta_{mn} \delta_{m'n'}$ and $\tr[O_{imn}^{jm'n'} \mbb{F}^{SS'}] = \delta_{ij} \delta_{mn'} \delta_{m'n} = 0$, we have
\begin{equation}
O_{imn}^{jm'n'}
=
\frac{\bigl( d_S  \mbb{I}^{SS'}  - \mbb{F}^{SS'}  \bigr)  \delta_{mn} \delta_{m'n'}}{d_S(d_S^2-1)}, \label{Eq:Ocomp}
\end{equation}
for $i\neq j$, where $d_S = 2^{N+k}$ is the dimension of the system $S$.

Substituting~\cref{Eq:crrrrrrleee,Eq:Ocomp} into~\cref{Eq:E}, it follows that
\begin{equation}
    E =
    2^k(2^k-1)\frac{2^{2(N+k)-\ell}-2^{\ell}}{2^{2(N+k)}-1} \tr\bigl[ (\xi^{B_{\rm rad}})^2 \bigr].
\end{equation}
Note that, since $\xi^{B_{\rm in}B_{\rm rad}}$ is a pure state, $\tr\bigl[ (\xi^{B_{\rm rad}})^2 \bigr] = \tr\bigl[ (\xi^{B_{\rm in}})^2 \bigr]  = 2^{-H_2(B_{\rm in})_{\xi}}$, where $H_2(B_{\rm in})_{\xi}=-\log \tr[(\xi^{B_{\rm in}})^2]$ is the collision entropy of the state $\xi^{B_{\rm in}}$.
Hence,~\cref{Eq:c34tpgrme} is simplified as
\begin{align}
&\EH \bigl[ \Delta_{{\rm cl}, W}(M_{{\rm pPGM},W}|\xi, U) \bigr] \notag \\
&\leq 
\frac{2^{2(\ell_{\rm th}-\ell)}}{1-\epsilon}
\biggl( 1- \frac{1}{2^{k}} \biggr)
\frac{1-2^{-2(N+k-\ell)}}{1-2^{-2(N+k)}} + 2^k P_{\xi}(\epsilon),\\
&\leq 
\frac{2^{2(\ell_{\rm th}-\ell)}}{1-\epsilon} + 2^k P_{\xi}(\epsilon),\label{Eq:nearfin}
\end{align}
where $\ell_{\rm th} = k + \frac{N-H_2(B_{\rm in})_{\xi}}{2}$.

We finally evaluate the second term $2^k P_{\xi}(\epsilon)$. 
We have
\begin{multline}
\log \bigl[2^k P_{\xi}(\epsilon) \bigr]\\
=
k + 2^{N+k-\ell+1} \biggl(N + k - \ell+ \log \frac{5}{\epsilon} \biggr) - \frac{2^{\ell} r c^2}{6} \log e.
\end{multline}
Using $\log r \geq H_2(B_{\rm in})_{\xi}$, it holds that
\begin{multline}
\log \bigl[2^k P_{\xi}(\epsilon) \bigr]
\leq 
k + 2^{N+k-\ell+1} \biggl(N + k - \ell+ \log \frac{5}{\epsilon} \biggr) \\
- \frac{c^2 \log e}{6}2^{\ell + H_2(B_{\rm in})_{\xi}}.
\end{multline}
The right-hand side is exactly the form of $\log \delta$ in the statement.
We hence obtain the desired conclusion:
\begin{equation}
\EH \bigl[ \Delta_{{\rm cl}, W}(M_{{\rm pPGM},W}|\xi, U) \bigr] 
\leq 
\frac{2^{2(\ell_{\rm th}-\ell)}}{1-\epsilon} + \delta.
\end{equation}

The second statement about the average error $\EH[\Delta_{\rm{q}}(\cD_{\rm CtoQ}|\xi, U)]$ on decoding quantum information simply follows by combining the first statement with~\cref{Cor:CtoQMUB}.
\end{proof}

\subsection{Proof of~\cref{Thm:HPformal}} \label{SS:owari}

\cref{Thm:HPformal} can be proved based on~\cref{Thm:result3}. 

\begin{proof}[Proof of~\cref{Thm:HPformal}]
Since the decoding errors are trivially zero when $N_n+k_n-\ell_n = 0$, we below consider the case when $N_n+k_n-\ell_n \geq 1$. In this case,~\cref{Eq:ccccceeer} is rewritten as 
\begin{equation}
    \lim_{n\to\infty} \bigl(k_n + 2(\ell_n - \ell_{\mathrm{th},n} ) - \log(N_n+k_n-\ell_n) \bigr) = \infty. \label{Eq:cmoervre}
\end{equation}

There exists an integer $n_0$ such that $2(1-\Lambda_{\xi_n}) <1$ and $k_n+ 2( \ell_n -\ell_{{\rm th},n}) >0$ for all $n\ge n_0$.
For $n\ge n_0$, we take a sequence $\{\epsilon_n\}$ that satisfies $\epsilon_n \in (2 (1-\Lambda_{\xi_n}), 1]$,
$\lim_{n \rightarrow \infty} \epsilon_n=\epsilon$, and
\begin{equation}
    \lim_{n \rightarrow \infty}
    \frac{\log\log \frac{1}{\epsilon_n}}{k_n +2(\ell_n - \ell_{{\rm th}, n}) - \log(N_n + k_n -\ell_n)}<1. \label{Eq:cni34vrera}
\end{equation}
Note that the denominator diverges when $n \rightarrow \infty$ due to the assumption.
We then apply~\cref{Thm:result3} to obtain
\begin{equation}
    \overline{\Delta_{{\rm cl},n}} \leq L_n + \delta_n \label{Eq:ionowkm}
\end{equation}
for $n\ge n_0$, where $L_n=\frac{2^{2(\ell_{{\rm th},n}-\ell_n)}}{1- \epsilon}$ and $\delta_n$ is given by
\begin{multline}
    \log \delta_n =k_n +2^{N_n +k_n -\ell_n +1}\biggl( N_n +k_n -\ell_n + \log \frac{5}{\epsilon_n} \biggr)\\
    - \frac{\log 2}{24} 2^{\ell_n + H_2(B_{\rm in})_{\xi_n}}.
\end{multline}
Note that we have used an obvious bound $c_n \leq 1/2$ for all $n$.

By taking the limit of~\cref{Eq:ionowkm}, it follows that
\begin{equation}
    \limsup_{n\rightarrow \infty} \biggl\{\frac{\overline{\Delta_{{\rm cl},n}}}{L_n} \biggr\} \leq 1 + \limsup_{n\rightarrow \infty} \frac{\delta_n}{L_n}. \label{Eq:korekorekorekore}
\end{equation}
In the following, we show that $\limsup_{n\rightarrow \infty} \delta_n/L_n = 0$.
We define $p_n$, $q_n$, and $r_n$ by
\begin{align}
    &p_n =\ell_n + H_2(B_{\rm in})_{\xi_n} = N_n + 2(k_n - \ell_{{\rm th},n}),\\
    &q_n =\log \bigl(k_n +2(\ell_n - \ell_{{\rm th}, n}) \bigr),\\
    &r_n =N_n +k_n -\ell_n + \log (N_n +k_n -\ell_n + \log \frac{5}{\epsilon_n}).
\end{align}
Using $p_n$, $q_n$, and $r_n$, we can write $\log \delta_n/L_n$ as
\begin{multline}
    \log \frac{\delta_n}{L_n} =
    2^{q_n} + 2^{1 + r_n} - \frac{c^2 \log 2}{6} 2^{p_n}   +\log(1-\epsilon_n).
\end{multline}
We below show that, when $n  \rightarrow \infty$, 
\begin{align}
    p_n - q_n - r_n &= p_n\bigl( 1 - \frac{q_n}{p_n} \bigr) - r_n \rightarrow \infty. \label{Eq:limtotyu}
\end{align}

It follows that 
\begin{align}
    \frac{q_n}{p_n} 
    &= \frac{\log \bigl(k_n +2(\ell_n - \ell_{{\rm th}, n}) \bigr)}{\ell_n + N_n + 2(k_n - \ell_{{\rm th}, n})}\\
    &\leq 
    \frac{\log \bigl(k_n +2(\ell_n - \ell_{{\rm th}, n}) \bigr)}{k_n + \ell_n - \ell_{{\rm th}, n}}\\
    &\leq 
    2 \frac{\log \bigl(k_n +2(\ell_n - \ell_{{\rm th}, n}) \bigr)}{k_n + 2(\ell_n - \ell_{{\rm th}, n})},
\end{align}
where the second line follows from $N_n + k_n \geq \ell_{{\rm th}, n}$ since $H_2(B_{\rm in})_{\xi_n} \geq 0$, and the last from $k_n \geq 0$.
This converges to zero when $n\rightarrow \infty$ since~\cref{Eq:cmoervre} implies that $k_n + 2(\ell_n - \ell_{{\rm th}, n}) \rightarrow \infty$ when $n\rightarrow \infty$. Hence, we have 
\begin{equation}
    \limsup_{n \rightarrow \infty }\frac{q_n}{p_n} = 0. \label{Eq:-1}
\end{equation}

We also have 
\begin{multline}
    p_n - r_n = k_n + 2(\ell_n - \ell_{{\rm th},n}) \\
    -\log\bigl(N_n +k_n -\ell_n + \log \frac{5}{\epsilon_n}\bigr).
\end{multline}
Since $N_n +k_n -\ell_n \geq 1$, we use $\log(x+y) \leq 1 + \log x  + \log y$ for any $x, y \geq 1$ to obtain
\begin{multline}
    p_n - r_n \geq 
    k_n + 2(\ell_n - \ell_{{\rm th},n}) -\log\bigl(N_n +k_n -\ell_n\bigr) \\
    - \log\log \frac{5}{\epsilon_n} - 1.
\end{multline}
The right-hand side diverges when $n \rightarrow \infty$ since the first three terms diverges due to the assumption (\cref{Eq:cmoervre}) and we have set $\epsilon_n$ such that~\cref{Eq:cni34vrera} is satisfied.
Hence, 
\begin{equation}
    \lim_{n \rightarrow \infty } (p_n - r_n)  = \infty. \label{Eq:-2}
\end{equation}

From~\cref{Eq:-1,Eq:-2}, we obtain~\cref{Eq:limtotyu}, which further implies that $\limsup_{n \rightarrow \infty} \log (\delta_n/L_n ) = -\infty$. Thus,
\begin{equation}
    \limsup_{n \rightarrow \infty} \frac{\delta_n}{L_n} = 0.
\end{equation}
Substituting this into~\cref{Eq:korekorekorekore}, we arrive at~\cref{Eq:ccccc}.

The statement about the error on decoding quantum information,~\cref{Eq:qqqqqqqq}, is similarly obtained.
\end{proof}

\section{Conclusions and outlooks} \label{S:COs}
In this paper, we have begun with a decoding circuit of a CSS code. The circuit is based on the two classical decoders of the linear codes that define the CSS code. 
We have then shown that the decoding circuit can be extended to handle a general QECC by replacing the two classical decoders with decoding measurements of the associated CQ codes. The decoding error is given by the two classical decoding errors of the CQ codes and by the degree of complementarity of the bases that define the CQ codes. 
We have also shown that the constructed decoding circuit, i.e., the C-to-Q decoder, is nearly optimal and capacity-achieving, when the decoding measurements of the CQ codes are suitably chosen. The power of the C-to-Q decoder has been demonstrated by applying it to the decoding problem of the Hayden-Preskill protocol. We have improved the previous results and revealed that the dynamics of black holes may be an optimal encoder for quantum information but is a poor encoder for classical information.

Since QEC is a key technique toward large-scale quantum information processing, we expect that the C-to-Q decoder has a number of applications to large-scale quantum communication, fault-tolerant quantum computation, and quantum cryptography.
This is analogous to the Petz recovery map~\cite{P1986}, which was originally proposed in a special context but later found many applications~\cite{BK2002,BDL2016,HJPW2004,FR2015,PSSY2022} in quantum information tasks as well as in theoretical physics.
It is also interesting to investigate the relations between the C-to-Q decoder and existing general constructions of decoders~\cite{P1986,GLMQW2022,BK2002,BDL2016,K2007, BR2009, T2010, Renes2016,UN2024}.

Apart from applications, it is of fundamental interest to further investigate the role of complementarity in QEC. Our construction of the C-to-Q decoder quantitatively shows that, if two types of classical information defined in two \emph{complementary} bases can be decoded with small errors, so does quantum information. This observation was previously made in an implicit manner~\cite{K2007, BR2009, T2010, Renes2016}.
We conjecture that this is true even if the two bases are not complementary.
More precisely, we expect that, if the two bases do not share a common strict subspace, then decoding classical information defined in such two bases is equivalent to decoding quantum information.
Here, we mean by `two bases do not share a common strict subspace' that any subspace spanned by a strict subset of one basis cannot be spanned by a strict subset of the other basis.
Proving this conjecture will contribute to better understand the relation between QEC and complementarity.

\section*{Acknowledgments}
Y.~N.\ is supported by JSPS KAKENHI Grant Number JP22K03464, by MEXT-JSPS Grant-in-Aid for Transformative Research Areas (A) ``Extreme Universe'' Grant Numbers JP21H05182 and JP21H05183, by JST CREST Grant Number JPMJCR23I3, and by JST PRESTO Grant Numbers JPMJPR1865 and JPMJPR2456.
T. M. is supported by JSPS KAKENHI Grant Number 21J12744.
This work is also supported by JST, Moonshot R$\&$D Grant Number JPMJMS2061.

\bibliography{Bib2}

\appendix

\section{Proof of~\cref{Prop:MESdiamondUnitary}} \label{App:MESdiamondUnitary}

We here prove~\cref{Prop:MESdiamondUnitary}. The statement is as follows.

\MESdiamondUnitary*

\begin{Proof}
    Since $\cM^A_{\rm twir}:=\mbb{E}_U\bigl[ \cD_{U^{\dagger}}^{C \rightarrow A} \circ \cN^{B \rightarrow C} \circ \cE_{U}^{A \rightarrow B}\bigr]$ is Hermitian-preserving, it follows from~\cref{Def:DiamondNorm} that
    \begin{equation}
        \bigl\| \id^A - \cM^A_{\rm twir} \bigr\|_{\diamond}
        =
        \max_{\ket{\psi}^{AR}} \bigl\| \ketbra{\psi}{\psi}^{AR} - \cM^A_{\rm twir}(\ketbra{\psi}{\psi}^{AR}) \bigr\|_1,
    \end{equation}
    where $R$ has the same dimension as $A$.
    In the following, without loss of generality, we consider a discrete unitary $1$-design $\{p_j, U_j\}_j$\footnote{When the probability distribution $\{p_j\}_j$ is not uniform, the ensemble is called a \emph{weighted unitary $1$-design}. Since our proof works both for a unitary $1$-design and a weighted unitary $1$-design, we do not explicitly distinguish them.}.
    Using the triangle inequality and the invariance of the trace norm under unitary, we have
    \begin{multline}
        \bigl\| \id^A - \cM^A_{\rm twir} \bigr\|_{\diamond}\\
        \leq
        \max_{\ket{\psi}^{AR}} \sum_j \bigl\| p_j\bigl(U_j^A \ketbra{\psi}{\psi}^{AR} U_j^{A \dagger} - \cM^A(U_j^A \ketbra{\psi}{\psi}^{AR} U_j^{A \dagger})\bigr) \bigr\|_1, \label{Eq:diamondupperoafwe}
    \end{multline}
    where $\cM^A := \cD^{C \rightarrow A} \circ \cN^{B \rightarrow C} \circ \cE^{A \rightarrow B}$. 
    
    We now use the fact that there exists a pure state $\ket{\Psi}^{ARC}$ on a system $ARC$ and a POVM $\{ \Gamma_j^C \}_j$ on $C$ such that 
    \begin{equation}
        \tr_C \bigl[ \Gamma_j^C \ketbra{\Psi}{\Psi}^{ARC}  \bigr] = p_j U_j^A \ketbra{\psi}{\psi}^{AR} U_j^{A \dagger}.    \label{Eq:awfe;oiregg}
    \end{equation}
    By summing up over $j$, we have
    \begin{equation}
        \Psi^{AR} = \sum_j p_j U_j^A \ketbra{\psi}{\psi}^{AR} U_j^{A \dagger} = \pi^A \otimes \psi^R. 
    \end{equation}
    Here, we used the property $\sum_j \Gamma_j^C = I^C$ of the POVM $\{ \Gamma_j^C\}_j$ and also the fact that $\{p_j, U_j\}_j$ is a unitary $1$-design. Further tracing out $R$, we obtain $\Psi^{A} = \pi^A$.
    Since the maximally entangled state $\ket{\Phi}^{AR}$ has the same marginal state $\Phi^A = \pi^A$, there exists an isometry $V^{R \rightarrow RC}$ such that
    \begin{equation}
        \ket{\Psi}^{ARC} = V^{R \rightarrow RC} \ket{\Phi}^{AR}.  \label{Eq:VPsidia}
    \end{equation}
    
    Consider a quantum channel $\chi^{R\rightarrow RX }$ defined by 
    \begin{multline}
        \chi^{R\rightarrow RX } (\rho^R)\\
        =\sum_j
        \tr_C (\Gamma_j^C V^{R\rightarrow RC } \rho^R V^{R\rightarrow RC \dagger}) \otimes  \ketbra{e_j}{e_j}^X,
    \end{multline}
    where $\{\ket{e_j}^X\}_j$ is an orthonormal basis in an ancillary system $X$. From~\cref{Eq:awfe;oiregg,Eq:VPsidia}, we have
    \begin{equation}
        \chi^{R\rightarrow RX } ( \ketbra{\Phi}{\Phi}^{AR} ) = \sum_j p_j  U_j^A  \ketbra{\psi}{\psi}^{AR} U_j^{A\dagger} \otimes \ketbra{e_j}{e_j}^X.
    \end{equation}
    This leads to
    \begin{widetext}
        \begin{align}
            \sum_j \bigl\|p_j \bigl( U_j \ketbra{\psi}{\psi}^{AR} U_j^{\dagger} - \cM^A(U_j^A \ketbra{\psi}{\psi}^{AR} U_j^{\dagger}) \bigr) \bigr\|_1
            &=   \bigl\| \sum_jp_j \bigl(U_j \ketbra{\psi}{\psi}^{AR} U_j^{\dagger} - \cM^A(U_j^A \ketbra{\psi}{\psi}^{AR} U_j^{\dagger}) \bigr) \otimes \ketbra{e_j}{e_j}^X \bigr\|_1\\
            &=   \bigl\| \chi^{R\rightarrow RX } \bigl(\ketbra{\Phi}{\Phi}^{AR} - M^A (\ketbra{\Phi}{\Phi}^{AR})  \bigr) \bigr\|_1\\
            &\leq
            \bigl\| \ketbra{\Phi}{\Phi}^{AR} - \cM^A(\ketbra{\Phi}{\Phi}^{AR}) \bigr\|_1, \label{Eq:o;awn4g4g4}
        \end{align}
    \end{widetext}
    where the last line follows due to the monotonicity of the trace distance under a quantum channel $\chi^{R\rightarrow RX }$.

    As~\cref{Eq:o;awn4g4g4} holds for any $\ket{\psi}^{AR}$, together with~\cref{Eq:diamondupperoafwe}, it holds that
    \begin{equation}
        \bigl\| \id^A - \cM^A_{\rm twir} \bigr\|_{\diamond}
        \leq
        \bigl\| \ketbra{\Phi}{\Phi}^{AR} - \cM^A(\ketbra{\Phi}{\Phi}^{AR}) \bigr\|_1.
    \end{equation}
    This is the statement of~\cref{Prop:MESdiamondUnitary} as $\cM^A_{\rm twir}=\mbb{E}_U\bigl[ \cD_{U^{\dagger}}^{C \rightarrow A} \circ \cN^{B \rightarrow C} \circ \cE_{U}^{A \rightarrow B}\bigr]$ and $\cM^A = \cD^{C \rightarrow A} \circ \cN^{B \rightarrow C} \circ \cE^{A \rightarrow B}$.    $\hfill \qed$
\end{Proof}

\section{The C-to-Q decoder for a CSS code} \label{App:CSS}
In this section, we explain how an $[[N,k]]$ CSS code can be decoded by the circuit proposed in~\cref{SSS:DecodeCSS}.

Let $C_1$ and $C_2$ be classical linear $[[N,k_1]]$-code and $[[N,k_2]]$-code, respectively, with $k_1-k_2=k$ and $C_2\subset C_1$.  Let $G_1$ and $G_2$ be generator matrices and $H_1$ and $H_2$ be parity-check matrices for the codes $C_1$ and $C_2$, respectively, all of which are binary matrices with $N$ columns (we will thus use the row vector convention).  Due to the condition $C_2\subset C_1$, we have $G_2 H_1^{\rm t}=0$, where the superscript ${\rm t}$ denotes the transposition.  It is then possible to take $G_1$ and $H_2$ as $G_1^{\rm t} = \begin{pmatrix}G_2^{\rm t} & G^{\rm t}\end{pmatrix}$ and $H_2^{\rm t} = \begin{pmatrix}H_1^{\rm t} & H^{\rm t}\end{pmatrix}$ such that $GH^{\rm t}=I_k$ holds.  (The reason why we can take these is as follows.  Since no nonzero linear combination of rows of $G$ belongs to $C_2$, $G H_2^{\rm t}=\begin{pmatrix}0 & GH^{\rm t}\end{pmatrix}$ must be full rank.  Thus, redefining $G$ as $(GH^{\rm t})^{-1}G$ will satisfy the condition.)  
Using these matrices, we can write an orthonormal basis $\{\ket{\bar{j}_Z}\}_{j=1,\ldots ,k}$ of the CSS code from 
$C_1$ and $C_2$ as the $N$-qubit states given by
\begin{equation}
    \ket{\bar{j}_Z}=2^{-k_2/2}\sum_{w\in\{0,1\}^{k_2}}\ket{(j G + w G_2)_Z}, \label{eq:encoded_state}
\end{equation}
which we call the logical $Z$ basis and corresponds to the basis $E$ in~\cref{SSS:a}.

Under the presence of an error in the $Z$ basis, represented by an $N$-bit string $e_1$, consider a classical decoder $f_1(z)$ for identifying the $Z$ basis index $j$ from an $N$-bit string $z=jG+wG_2+e_1$ of the $Z$-basis measurement outcome. Assuming that no knowledge is available on the correlation between the error $e_1$ and the input value $jG+wG_2$, the random sequence $w$ does not help in identifying $j$.  Hence, we may assume that the decoder ignores the $wG_2$ term, namely, 
\begin{equation}
    f_1(z+w G_2) = f_1(z).
\end{equation}
We further assume that the decoder is symmetric over the change in the encoded value $j$, namely,
\begin{equation}
    f_1(z+wG_2+jG) = f_1(z) + j.
\end{equation}
Then, an error $e_1$ is perfectly correctable by the decoder $f_1(z)$ if and only if $f_1(e_1)=0$ since $f_1(jG+wG_2+e_1)=j+f_1(e_1)$.

The logical $X$ basis $\{\ket{\bar{l}_X}\}_{l=1,\ldots,k}$ corresponding to the $F$ basis in~\cref{SSS:a} is defined as
\begin{align}
    \ket{\bar{l}_X} &:= 2^{-k/2}\sum_{j \in\{0,1\}^k}(-1)^{j l^{\rm t}}\ket{\bar{j}_Z} \\
    &= 2^{(N-k_1)/2}\sum_{x:\,x G^{\rm t}=l,\,x G_2^{\rm t}=0}\ket{x_X}, \\
    &= 2^{(N-k_1)/2}\sum_{v\in\{0,1\}^{N-k_1}} \ket{(l H + v H_1)_X}
\end{align}
where $\ket{x_X}$ is an $X$-basis $N$-qubit state given by
\begin{equation}
    \ket{x_X}=2^{-N/2}\sum_{z\in\{0,1\}^N}(-1)^{zx^{\rm t}}\ket{z_Z}.
\end{equation}
We similarly consider a classical decoder for identifying the $X$ basis index $l$ represented by $f_2(x)$, which satisfies $f_2(x+v H_1+lH)=f_2(x) + l$.  An $X$-basis error $e_2$ is perfectly correctable by the decoder $f_2(x)$ if and only if $f_2(e_2)=0$. 

We now turn to check that the circuit depicted in~\cref{Fig:CSSCodes} actually works as a decoder for this CSS code.  Let $\ket{\bar{\psi}}:=\sum_j \alpha_j \ket{\bar{j}_Z}$ be an arbitrary encoded state of this CSS code and assume that a Pauli error $X^{e_1}Z^{e_2}$ has occurred on this encoded state, where $X^y:=X^{y_1}\otimes \cdots \otimes X^{y_N}$ for an $N$-bit string $y=(y_1,\ldots,y_N)$. 
Then, by applying the isometry $R_Z$ given in~\cref{eq:z_basis_isometry}, we have
\begin{align}
    &R_Z X^{e_1}Z^{e_2}\ket{\bar{\psi}} \nonumber \\
    \begin{split}
    &= 2^{-k_2/2}\sum_{j\in\{0,1\}^k}\alpha_j \sum_{w\in\{0,1\}^{k_2}}(-1)^{(jG+wG_2)e_2^{\rm t}} \\
    & \hspace{3cm}\ket{(jG + wG_2 + e_1)_Z}\ket{(j+f_1(e_1))_Z} 
    \end{split}\\
    &= 2^{(N-k_2)/2}\sum_{j}\alpha_j \sum_{x:\,xG_2^{\rm t} =e_2 G_2^{\rm t}} (-1)^{jG(x^{\rm t} + e_2^{\rm t})+e_1x^{\rm t}}\nonumber\\
    & \hspace{4.5cm} \ket{x_X}\ket{(j+f_1(e_1))_Z}.
\end{align}
We then apply the $X$-basis measurement on the $N$-qubit system to obtain an outcome string $x$ and perform a correction operation $Z^{f_2(x)}$ to the $k$-qubit system as depicted in~\cref{Fig:CSSCodes}. Since the outcome $x$ satisfies $xG_2^{\rm t} = e_2 G_2^{\rm t}$, it is written in the form of $x = lH+vH_1+e_2$. The state after the error correction is then given by
\begin{align}
    &\sum_{j}\alpha_j (-1)^{jG(x^{\rm t} + e_2^{\rm t})+e_1x^{\rm t}} Z^{f_2(x)}\ket{(j+f_1(e_1))_Z} \\
    &= \sum_{j}\alpha_j (-1)^{jl^{\rm t} + e_1 x^{\rm t}+ (j+f_1(e_1)) (l+f_2(e_2))^{\rm t}}\ket{(j+f_1(e_1))_Z} \\
    &= (-1)^{e_1 x^{\rm t}+ f_1(e_1) l^{\rm t}}Z^{f_2(e_2)}X^{f_1(e_1)}\ket{\psi}.
\end{align}
Therefore, the quantum decoder succeeds if $f_1(e_1)=0$ and $f_2(e_2)=0$ are satisfied, i.e., if classical decoders for $Z$ and $X$ succeed, which is what we expected.

The decoding error can be easily obtained in this case.  Let $\Delta_{{\rm cl}, Z}$ and $\Delta_{{\rm cl}, X}$ be the failure probability of the classical decoders $f_1$ and $f_2$, respectively.
Then, the fidelity between the original maximally entangled state $\ket{\Phi}$ and the output of the noisy channel followed by the decoder in~\cref{Fig:CSSCodes} is no smaller than $1 - \Delta_{{\rm cl}, Z} - \Delta_{{\rm cl}, X}$ since the output state is $\ket{\Phi}$ when $f_1(e_1)=f_2(e_2)=0$ and it is orthogonal to $\ket{\Phi}$ otherwise.  Thus, when measured in the trace distance, the decoding error is bounded from the above by $\sqrt{\Delta_{{\rm cl}, Z} + \Delta_{{\rm cl}, X}}$ as claimed in~\cref{eq:error_for_css}.  

This error bound is better than the one in~\cref{Cor:CtoQMUB} for the following reason. 
Since CSS codes are designed so that the $Z$- and $X$-error correction can be done independently, the decoding error of the $X$-error correction is independent of whether or not the preceding $Z$-error correction succeeds.  
This is not the case for general codes: decoding errors for $Z$ and $X$ may correlate, and thus we need to use a trick to decouple these correlations, which results in a looser bound.

\section{Derivation of~\cref{Eq:koregeri}} \label{App:last}
In this section, we provide an upper bound of $\Xi_{EF}$, which is defined as
\begin{align}
\Xi_{EF} =  1 - \sum_{l=0}^{d-1} \tr \bigl[\cT^{A \rightarrow C}(\pi^A) M_{F, l}\bigl] F_{BC}\bigl({\rm unif}_d, p_l\bigr),
\end{align}
where $\cT^{A \rightarrow C} = \cN^{B\rightarrow C} \circ \cE^{A \rightarrow B}$.

For simplicity, we denote $F_{BC}\bigl({\rm unif}_d, p_l\bigr)$ by $F_l$, and $\tr[\cT^{A \rightarrow C}(\ketbra{m_F}{m_F}^A)M_{F, l}]$ by $q(l,m)$. 
Since $\pi$ is the completely mixed state, we have
\begin{align}
    \Xi_{EF} &= 1 - \frac{1}{d}\sum_{l,m} q(l,m) F_l\\
    &\leq 1 - \frac{1}{d}\sum_{l,m} q(l,m) F_m - \frac{1}{d}\sum_{l,m} q(l,m) (F_l-F_m). \label{Eq:crbrrt}
\end{align}
Using the facts that $\sum_l q(l,m) = 1$ for any $m$ and that
\begin{align}
    \sum_{l,m} q(l,m) (F_l-F_m)
    &=
    \sum_{l\neq m} q(l,m) (F_l-F_m)\\
    &\geq
    (F_{BC, \min}-F_{BC, \max}) \sum_{l\neq m} q(l,m) \\
    &=(F_{BC, \min}-F_{BC, \max}) d \Delta_{{\rm cl}, F},
\end{align}
where $F_{BC, \max} = \max_l F_{BC}\bigl({\rm unif}_d, p_l\bigr)$ and $F_{BC, \min} = \min_l F_{BC}\bigl({\rm unif}_d, p_l\bigr)$,
we obtain
\begin{align}
    \Xi_{EF} 
    &\leq 1 - \frac{1}{d}\sum_{m} F_m + (F_{BC, \max}-F_{BC, \min}) \Delta_{{\rm cl}, F}.
\end{align}

\section{Derivation of~\cref{Eq:HN}} \label{App:Holevo}

We here derive~\cref{Eq:HN}, namely, 
\begin{align}
&\Delta_{{\rm cl}, Z} (M_{\rm pPGM}|\cN \circ \cE) \leq \frac{1}{d} \sum_{i \neq j} \tr \bigl[ \Pi_i^C \tau_j^C \bigr].
\end{align}
Here, $M_{\rm pPGM}$ is a pPGM given by
\begin{equation}
    M_{\rm pPGM} = \bigl\{ M_j = \Pi^{-1/2} \Pi_j \Pi^{-1/2} \bigr\}_j,
\end{equation}
where $\Pi_j$ is a projection onto the support of $\tau_j^C := \cT^{A \rightarrow C}(\ketbra{j}{j}^A)$ with $\cT^{A \rightarrow C}$ being $\cN^{B\rightarrow C} \circ \cE^{A \rightarrow B}$, and $\Pi := \sum_j \Pi_j$. 
In the following, we omit the superscript $C$.

The proof is based on the same idea as~\cite{H1998,SW1998} except that we do not consider the asymptotic situation. For this reason, we need a slight modification.

Using a diagonalization of $\tau_j$ such as $\tau_j = \sum_{z = 0}^{t_j-1} \lambda_z^{(j)} \ketbra{\varphi_z^{(j)}}{\varphi_z^{(j)}}$, where $t_j = {\rm rank} \tau_j$, 
the projection $\Pi_j$ is given by $\tau_j = \sum_{z = 0}^{t_j-1} \ketbra{\varphi_z^{(j)}}{\varphi_z^{(j)}}$, which leads to 
\begin{equation}
    \tr \bigl[ M_j \tau_j \bigr]
    =
    \sum_{z,w=0}^{t_j-1} \lambda_z^{(j)} \bigl| \bra{\varphi^{(j)}_z} \Pi^{-1/2} \ket{\varphi_w^{(j)}} \bigr|^2.
\end{equation}
Since $\Delta_{{\rm cl}, Z} (M_{\rm pPGM}|\cN \circ \cE) = 1 -  d^{-1}\sum_j \tr [ M_j \tau_j ]$, we have
\begin{align}
    &\Delta_{{\rm cl}, Z} (M_{\rm pPGM}|\cN \circ \cE) \notag \\
    &=
    \frac{1}{d} \sum_{j=0}^{d-1} \sum_{z=0}^{t_j-1} \lambda_z^{(j)} \biggl( 1 -  \sum_{w=0}^{t_j-1} \bigl| \bra{\varphi^{(j)}_z} \Pi^{-1/2} \ket{\varphi_w^{(j)}} \bigr|^2 \biggr),\\
    & \leq
    \frac{1}{d} \sum_{j=0}^{d-1} \sum_{z=0}^{t_j-1} \lambda_z^{(j)} \biggl( 1 - \bigl| \bra{\varphi^{(j)}_z} \Pi^{-1/2} \ket{\varphi_z^{(j)}} \bigr|^2 \biggr),\\
    & \leq
    \frac{2}{d} \sum_{j=0}^{d-1} \sum_{z=0}^{t_j-1} \lambda_z^{(j)} \biggl( 1 - \bra{\varphi^{(j)}_z} \Pi^{-1/2} \ket{\varphi_z^{(j)}}  \biggr),
\end{align}
where the first equality is due to $\sum_{z=0}^{t_j-1} \lambda_z^{(j)}=1$ for any $j$, and the last line to $1-x^2 \leq 2(1-x)$.

Define two matrices, $\Gamma$ and $\Lambda$, whose matrix elements are given by
\begin{align}
    &\Gamma_{(j,z), (i,w)} := \braket{\varphi_z^{(j)}}{\varphi_w^{(i)}},\\ &\Lambda_{(j,z), (i,w)} := \delta_{(j,z), (i,w)} \lambda_z^{(j)}.
\end{align}
Note that they are both positive semi-definite.
By a direct calculation, it is straightforward to check that $\bigl(\Gamma^{1/2} \bigr)_{(j,z), (i,w)} = \bra{\varphi_z^{(j)}} \Pi^{-1/2} \ket{\varphi_w^{(i)}}$. Hence, it follows that
\begin{align}
    \Delta_{{\rm cl}, Z} (M_{\rm pPGM}|\cN \circ \cE)
    &\leq
    \frac{2}{d} \tr \bigl[ \Lambda (I - \Gamma^{1/2}) \bigr].
\end{align}
Using the matrix inequality $2(I - \Gamma^{1/2}) \leq 2I - 3 \Gamma + \Gamma^2$, which follows from the inequality of the same form for non-negative real number $x \geq 0$, we obtain
\begin{align}
    &\Delta_{{\rm cl}, Z} (M_{\rm pPGM}|\cN \circ \cE) \notag \\
    &\leq
    \frac{1}{d} \tr \bigl[ \Lambda ( 2I - 3 \Gamma + \Gamma^2) \bigr]\\
    &=
    \frac{1}{d} \sum_{j=0}^{d-1} \sum_{z=0}^{t_j-1} \lambda^{(j)}_z \biggl( 2 - 3 \Gamma_{(j, z), (j, z)}  + \sum_{i=0}^{d-1} \sum_{w=0}^{t_i-1} |\Gamma_{(j, z), (i, w)}|^2  \biggr)\\
    &=
    \frac{1}{d} \sum_{j \neq i} \sum_{z, w} \lambda^{(j)}_z  |\Gamma_{(j, z), (i, w)}|^2\\
    &=
    \frac{1}{d} \sum_{j \neq i} \sum_{z, w} \lambda^{(j)}_z  | \braket{\varphi_z^{(j)}}{\varphi_w^{(i)}}|^2\\
    &=
    \frac{1}{d} \sum_{j \neq i} \tr [\Pi_i \tau_j],
\end{align}
where the third last line follows from the fact that $\Gamma_{(j,z),(j,w)} = \delta_{zw}$. \\

We finally comment on the fact that a similar statement follows from the Hayashi-Nagaoka's lemma~\cite{HN2003}, which states that $I - M_j \leq 2(I - \Pi_j) + 4 \sum_{i (\neq j)} \Pi_j$. Applying this, we immediately have
\begin{align}
    &\Delta_{{\rm cl}} (M_{\rm pPGM}|\cN \circ \cE) \leq \frac{4}{d} \sum_{i \neq j} \tr \bigl[ \Pi_i \tau_j \bigr], 
\end{align}
where we have used $\tr [\Pi_j \tau_j] = 1$. This is slightly looser than~\cref{Eq:HN} by a factor of $4$.

\end{document}